\documentclass{article}

     \usepackage[preprint,nonatbib]{neurips_2022}

\usepackage{todonotes}
\usepackage[utf8]{inputenc} 
\usepackage[T1]{fontenc}    
\usepackage{hyperref}       
\usepackage{url}            
\usepackage{booktabs}       
\usepackage{amsfonts}       
\usepackage{nicefrac}       
\usepackage{microtype}      
\usepackage{xcolor}         

\usepackage[ruled,vlined,linesnumbered]{algorithm2e}
\usepackage{latexsym}
\usepackage{amssymb}
\usepackage{amsmath}
\usepackage{amsthm}
\usepackage{thmtools,thm-restate}

\newtheorem{theorem}{Theorem}[section]
\newtheorem{lemma}[theorem]{Lemma}
\newtheorem{claim}[theorem]{Claim}
\newtheorem{corollary}[theorem]{Corollary}

\theoremstyle{definition}
\newtheorem{definition}[theorem]{Definition}
\theoremstyle{plain}

\newcommand{\bigO }{\mathcal{O}}
\newcommand{\eps}{\varepsilon}
\newcommand{\opt}{\texttt{OPT}}
\newcommand{\GRD}{\texttt{GRD}}

\newcommand{\poly}{\mathrm{poly}}

\newcommand{\Z}{\mathbb Z}

\newcommand{\Q}{\mathbb Q}

\newcommand{\U}{\mathcal{U}}

\newcommand{\sym}{\mathcal{S}}
\newcommand{\obj}{\texttt{Obj}}
\newcommand{\OPT}{\texttt{OPT}}
\newcommand{\co}{\texttt{cost}}

\newcommand{\AS}[1]{{#1}}

\newcommand{\sdr}[1]{}
\newcommand{\asr}[1]{}
\newcommand{\akr}[1]{}
\newcommand{\dcr}[1]{}

\newcommand{\newfontobj}[2]{
    \newcommand{#1}[1]{
\expandafter\def\csname##1\endcsname{{#2 ##1}}}}

\newfontobj{\class}{\rm} 

\class{EXP}
\class{NP}
\class{P}
\class{coNP}
\class{NL}

\title{Fair Rank Aggregation}

\author{
    Diptarka Chakraborty \\
    School of Computing \\
    National University of Singapore \\
    \texttt{diptarka@comp.nus.edu.sg} \\
    \And
    Syamantak Das \\
    Department of Computer Science and Engineering \\
    Indraprastha Institute of Information Technology, Delhi \\
    \texttt{syamantak@iiit.ac.in} \\
    \And
    Arindam Khan \\
    Department of Computer Science and Automation \\
    Indian Institute of Science, Bengaluru \\
    \texttt{arindamkhan@iisc.ac.in} \\
    \And
    Aditya Subramanian \\
    Department of Computer Science and Automation \\
    Indian Institute of Science, Bengaluru \\
    \texttt{adityasubram@iisc.ac.in}
}

\begin{document}
\maketitle

\begin{abstract}
    Ranking algorithms find extensive usage in diverse areas such as web search, employment, college
    admission, voting, etc.  The related rank aggregation problem deals with combining multiple
    rankings into a single aggregate ranking.  However, algorithms for both these problems might be
    biased against some individuals or groups due to implicit prejudice or marginalization in the
    historical data.  We study ranking and rank aggregation problems from a fairness or diversity
    perspective, where the candidates (to be ranked) may belong to different groups and each group
    should have a fair representation in the final ranking. We allow the designer to set the
    parameters that define fair representation. These parameters specify the allowed range of the
    number of candidates from a particular group in the top-$k$ positions of the ranking.  Given any
    ranking, we provide a fast and exact algorithm for finding the closest fair ranking for the
    Kendall tau metric under \textit{block-}fairness.
    We also provide an exact algorithm for finding the closest fair ranking for the
    Ulam metric under \textit{strict-}fairness, when there are only $\bigO(1)$ number of groups.  Our
    algorithms are simple, fast, and might be extendable to other relevant metrics. We also give a
    novel  meta-algorithm for the general rank aggregation problem under the fairness framework.
    Surprisingly, this meta-algorithm works for any generalized mean objective (including center and
    median problems) and any fairness criteria. As a byproduct, we obtain 3-approximation algorithms
    for both center and median problems, under both Kendall tau and Ulam metrics. Furthermore, using
    sophisticated techniques we obtain a $(3-\varepsilon)$-approximation algorithm, for a constant
    $\varepsilon>0$,  for the Ulam metric under strong fairness.
\end{abstract}

\section{Introduction}
Ranking a set of candidates or items is a ubiquitous problem arising in diverse areas ranging from
social choice theory~\cite{brandt2016handbook} to information retrieval~\cite{Harman92a}.  Given a
set of $d$ candidates and a set of $n$ different, potentially conflicting, rankings of these
candidates, one fundamental task is to determine a single ranking that best summarizes the
preference orders in the individual rankings.  This summarizing task, popularly termed \emph{rank
aggregation}, has been widely studied from a computational viewpoint over the last two
decades~\cite{dwork2001rank, fagin2003efficient, gleich2011rank, azari2013generalized}.  Most
well-studied rank aggregation paradigms are {\em median rank aggregation} (or simply \emph{rank
aggregation})~\cite{kemeny1959mathematics, young1988condorcet, young1978consistent, dwork2001rank}
and {\em maximum rank aggregation}~\cite{BACHMAIER20152, biedl2009complexity, popov2007multiple},
which are based on finding the {\em median} and {\em center} of the given set of rankings,
respectively.

Recently, fairness and diversity have become a natural prerequisite for ranking algorithms where
individuals are rated for access to goods and services or ranked for seeking facilities in education
(e.g., obtaining scholarship or admission), employment (e.g., hiring or promotion in a job), medical
(e.g., triage during a pandemic), or economic opportunities (e.g., loan lending). Some concrete
examples include university admissions through affirmative action in the
USA~\cite{deshpande2005affirmative} or the reservation system in jobs in
India~\cite{borooah2010social}, where we want rankings to be fair to mitigate the prevalent
disparities due to historical marginalization. Rankings not being fair may risk promoting extreme
ideology~\cite{costello2016views} or certain stereotypes about dominating/marginalized communities
based on sensitive attributes like gender or race~\cite{kay2015unequal, bolukbasi2016man}. There has
been a series of  works on fair ranking algorithms, see~\cite{zehlike2017fa, asudeh2019designing,
castillo2019fairness, garcia2021maxmin, gorantla2021problem, patro2022fair, zehlike2021fairness} and
the references therein.

A substantial literature on algorithmic fairness focuses on {\em group fairness} to facilitate {\em
demographic parity} \cite{dwork2012fairness} or {\em equal opportunity}~\cite{hardt2016equality}:
typically this is done by imposing fairness constraints which require that top-$k$ positions in the
ranking contain {\em enough} candidates from {\em protected} groups that are typically
underrepresented due to  prevalent discrimination (e.g., due to gender, caste, age, race, sex,
etc.).  In many countries, group fairness constraints are being enforced by legal
norms~\cite{EuroDiv, USDiv}.  For example, in Spain 40\% of candidates for elections in large voting
districts must be women~\cite{verge2010gendering}, in India 10\% of the total recruitment for civil
posts and services in government are reserved for people from Economically Weaker Society (EWS)
\cite{singh2019analysis}, etc.

In this paper, we study group fairness, more specifically proportional fairness (sometimes also
referred to as $p$-fairness \cite{baruah1996proportionate}).  Inspired by the {\em Disparate Impact}
doctrine, this notion of fairness mandates that the output of an algorithm must contain a fair
representation of each of the `protected classes' in the population.  In the context of ranking, the
set of candidates is considered to be partitioned into $g$ groups $G_1, G_2, \dots, G_g$. For each
group $G_i, i\in[g]$, we have two parameters $\alpha_i\in [0,1], \beta_i\in [0,1]$. A ranking
$\pi$ of the set of items is called {\em proportionally fair} if for
\AS{a given position $k\in [n]$}, and for every group $G_i$, the following two properties are satisfied: (a)
\emph{Minority Protection}: The number of items from group $G_i$, which are in the top-$k$ positions
$\pi(1), \pi(2), \dots, \pi(k)$, is at least $\lfloor \alpha_i\cdot k \rfloor$, and (b)
\emph{Restricted Dominance}: The number of items from group $G_i$, which are in the top-$k$
positions $\pi(1), \pi(2), \dots, \pi(k)$, is at most $\lceil \beta_i\cdot k \rceil$.

To compare different rankings several distance functions have been considered defined on the set of
permutations/rankings, such as  Kendall tau distance~\cite{kendall1938new, DG77,
kemeny1959mathematics, young1988condorcet, young1978consistent, dwork2001rank, ACN08,
kenyon2007rank, kumar2010generalized} (also called \emph{Kemeny distance} in case of rank
aggregation), Ulam distance~\cite{AD99, CMS01, CK06, AK10, AN10, NSS17, BS19, CDK21,
chakraborty2021approximating}, Spearman footrule distance~\cite{spearman1904proof,
spearman1906footrule, DG77, dwork2001rank, kumar2010generalized, BACHMAIER20152}, etc.  Among these,
Kendall tau distance is perhaps the most common measure used in ranking as it is the only known
measure to simultaneously satisfy several required properties such as neutrality, consistency, and
the extended Condorcet property~\cite{kemeny1959mathematics, young1988condorcet}.  The Ulam metric
is another widely-used measure in practice as it is also a simpler variant of the general edit
distance metric which finds numerous applications in computational biology, DNA storage system,
speech recognition, classification, etc. (e.g., see~\cite{CMS01, CDK21,
chakraborty2021approximating}).

One natural computational question related to fairness in ranking is, given a ranking, how to find
its closest fair ranking under $p$-fairness. Celis et al.~\cite{CelisSV18} considered this problem
and gave exact and approximation algorithms under several ranking metrics such as discounted
cumulative gain (DCG), Spearman footrule, and Bradley-Terry. However, their algorithms do not extend
to Kendall tau and Ulam metric, two of the most commonly used ranking metrics.

Fair rank aggregation is relatively less studied.  Recently, Kulman et al.~\cite{kuhlman2020rank}
initiated the study of fair rank aggregation  under Kendall tau metric. However, their fairness
notion is based on  {\em top-$k$ statistical parity} and {\em pairwise statistical parity}. These
notions are quite restricted. For example, their results only hold for binary protected attributes
(i.e., $g = 2$) and and does not satisfy $p$-fairness.  Informally, pairwise statistical parity
considers pairs of items from different groups in an aggregated manner and does not take into
account the actual rank of the items in the final ranking. See \cite{wei22} for an example on why
the fairness notion in \cite{kuhlman2020rank} does not satisfy $p$-fairness.

\subsection{Our Contributions.}
Our first main contribution is an exact algorithm for the closest fair ranking (CFR) problem (see
\AS{Definition~\ref{def:fair}}) for Kendall tau and Ulam metrics. For the
Kendall tau metric, we give the \emph{first exact algorithm} for the closest \AS{block-}fair ranking problem
(\autoref{thm:StrongGreedy}). Our algorithm is simple and based on a greedy strategy; however, the
analysis is delicate. It exploits the following interesting and perhaps surprising fact. Under the
Kendall tau metric, given a fixed (possibly unfair) ranking $\pi$, there exists a closest fair
ranking $\pi'$ to $\pi$ such that for every group $G_i, i\in[g]$, the relative ordering of
elements in $G_i$ remains unaltered in $\pi'$ compared to $\pi$ (\autoref{clm:RelOrderProp}).  Then,
for the {\em Ulam metric}, we give a polynomial time dynamic programming algorithm for the closest
\AS{strict-}fair ranking problem when the number of groups $g$ is a constant (\autoref{thm:ulamdp}). In
practice, the number of protected classes is relatively few, and hence our result gives an efficient
algorithm for such cases.

Our second significant contribution is the  study of \emph{rank aggregation problem under a
generalized notion of proportional fairness}.  One of our main contributions is to develop
anovel algorithmic framework for the fair rank aggregation that solves a wide variety of
rank aggregation objectives

satisfying such generic fairness constraints. An essential takeaway of our work is that a set of
potentially biased rankings can be aggregated into a fair ranking with only a small loss in the
`quality' of the ranking.  We study $q$-mean Fair Rank Aggregation (FRA), where given a set of
rankings $\pi_1, \pi_2, \dots, \pi_n$, a (dis)similarity measure (or distance function) $\rho$ between two
rankings, and any $q \geq 1$, the task is to determine a {\em fair} ranking $\sigma$ that minimizes
the generalized mean objective:  $\left( \sum_{i=1}^{n} \rho(\pi_i, \sigma)^q \right)^{1/q} $. We
would like to emphasize that in general, $q$-mean objective captures two classical data aggregation
tasks: One is \emph{median} which asks to minimize the sum of distances (i.e., $q=1$) and another is
\emph{center} which asks to minimize the maximum distance to the input points (i.e., $q=\infty$).

We show generic reductions of the $q$-mean Fair Rank Aggregation (FRA) to the problem of determining
the \emph{closest fair ranking} (CFR) to a given ranking. More specifically, we show that any
$c$-approximation algorithm for the closest fair ranking problem can be utilized as a blackbox to
give a $(c + 2)$-approximation to the FRA for any $q \geq 1$ (\autoref{thm:fairmeta}). This result
is oblivious to the specifics of the (dis)similarity measure and only requires the measure to be a
metric. Using the exact algorithms for CFR for the Kendall tau, Spearman footrule, and Ulam
metrics (for constantly many groups), we thus obtain  3-approximation  algorithms for the FRA
problem under these three (dis)similarity measures, respectively. Further, we provide yet another
simple algorithm that even breaks below 3-factor for the Ulam metric. For $q=1$, by combining the
above-stated 3-approximation algorithm with an additional procedure, we achieve a
$(3-\varepsilon)$-approximation factor (for some $\varepsilon > 0$) for the FRA under the Ulam, for
constantly many groups (\autoref{thm:Ulambetter}). We also provide another reduction from FRA to one
rank aggregation computation (without fairness) and a CFR computation (\autoref{thm:fairmeta2}), and
as a corollary get an $\bigO (d^3\log d+n^2d)$-time algorithm for Spearman footrule when $q=1$
(Corollary~\ref{cor:SFmed}). We summarize our main results in Table~\ref{tab:result}.

\begin{center}
    \begin{tabular}{|c||c|c|c|c|c|}
        \hline
        Problem & Metric & \#Groups & Approx Ratio & Runtime & Reference\\ \hline
        {\bf CFR} & Kendall tau (\AS{block-fair}) & Arbitrary & Exact & $\bigO (d^2)$& \autoref{thm:StrongGreedy} \\ \hline
                  & Ulam (\AS{strict-fair}) & Constant & Exact & $\bigO (d^{g+2})$ & \autoref{thm:ulamdp} \\ \hline
                  & Spearman footrule (\AS{strict-fair}) & Arbitrary & Exact & $\bigO (d^3\log d)$& \cite{CelisSV18}\\
                  \hline
                  \hline
        {\bf FRA} & Kendall tau (\AS{block-fair})& Arbitrary & 3 & $\bigO (nd^2+n^2d\log d)$ & Corollary~\ref{cor:KTAggr}  \\ \hline
                  & Ulam (\AS{strict-fair}) & Constant & $3$ & $\bigO (nd^{g+2}+n^2d\log d)$&Corollary~\ref{cor:UlamAggr} \\ \hline
                  & Ulam  (\AS{strict-fair}, $q=1$) & Constant & $3-\varepsilon$ & $\poly(n,d)$ &\autoref{thm:Ulambetter} \\ \hline
                  & Spearman footrule (\AS{strict-fair}) & Arbitrary & 3 & $\bigO (d^3\log d+n^2d+nd^2)$ & Corollary~\ref{cor:SFmed}\\
                  \hline
                  \hline

    \end{tabular}
    \label{tab:result}
\end{center}

{\bf Comparison with concurrent work.} Independently and concurrently to our work, Wei et al.~\cite{wei22} considers the
fair ranking problem under a setting that is closely related to ours. However, the fairness criteria
in their work are much more restrictive compared to ours as follows. Their algorithms for CFR are only
designed for a special case of our formulation where for each group $G_i$ and any position $k$ in
the output ranking, $\alpha_i = \beta_i = p(i)$, where $p(i)$ denotes the proportion of group $G_i$ in the entire population. Further, under the Kendall tau metric, they give a
polynomial time exact algorithm for CFR only for the special case of binary groups ($g=2$). They
also give additional algorithms for multiple groups - an exact algorithm that works in time
exponential in the number of groups and a polynomial time $2$-approximation. In contrast, we
\asr{can we still say that `we fully resolve CFR for KT'?}
resolve the CFR problem under the Kendall tau metric by giving a simple polynomial time algorithm for the case of multiple
groups and any arbitrary bounds on $\alpha_i$ and $\beta_i$ for each group $G_i$ \AS{under a notion of \textit{block fairnes}}. Further, we give the
first results for CFR and FRA under the Ulam metric as well.

\subsection{Other related work}
Recent years have witnessed a growing concern over ML algorithms or, more broadly, automated decision-making processes being biased~\cite{barocas2017fairness, mehrabi2021survey}. This bias or unfairness might stem both from implicit bias in the historical dataset or from prejudices of human agents who are responsible for generating part of the input.
Thus fair algorithms have received recent attention in machine learning and related communities. In particular, the notion of group fairness have been studied in classification~\cite{huang2019stable}, clustering~\cite{chierichetti2017fair}, correlation clustering~\cite{pmlr-v108-ahmadian20a},
resource allocation~\cite{Patel0L21}, online learning~\cite{patil2020achieving}, matroids and matchings~\cite{pmlr-v89-chierichetti19a}.
Specially, fair clustering problem is closely related to our problem. Fair rank aggregation can be considered as the fair 1-clustering problem where the input set is a set of rankings. See \cite{HuangJV19, ChenFLM19, BeraCFN19, BackursIOSVW19} for more related work on fair clustering.



\section{Preliminaries}
\paragraph*{Notations.}For any $n \in \mathbb{N}$, let $[n]$ denote the set $\{1,2,\dots,n\}$. We
refer to the set of all permutations/rankings over $[d]$ by $\sym_d$. Throughout this paper we
consider any permutation $\pi \in \sym_d$ as a sequence of numbers $a_1,a_2,\dots,a_d$ such that
$\pi(i)=a_i$, and we say that the \emph{rank} of $a_i$ is $i$. For any two $x,y \in [d]$ and a
permutation $\pi \in \sym_d$, we use the notation $x <_{\pi} y$ to denote that the rank of $x$ is
less than that of $y$ in $\pi$. For any subset $I=\{i_1 <i_2 <\cdots < i_r\}\subseteq [d]$, let
$\pi(I)$ be the sequence $\pi(i_1), \pi(i_2),\dots , \pi(i_r)$ (which is essentially a subsequence
of the sequence represented by $\pi$). When clear from the context, we use $\pi(I)$ also to denote
the set of elements in the sequence $\pi(i_1), \pi(i_2),\dots , \pi(i_r)$. For any $k \in [d]$ and
a permutation $\pi \in \sym_d$, we refer to $\pi([k])$ as the $k$-length \emph{prefix} of $\pi$. For
any prefix $P$, let $|P|$ denote the length of that prefix. For any two prefixes $P_1,P_2$ of a given string, we use $P_1 \subseteq P_2$ to denote $|P_1| \le |P_2|$.

\paragraph*{Distance measures on rankings.}There are different distance functions being considered
to measure the dissimilarity between any two rankings/permutations. Among them, perhaps the most
commonly used one is the \emph{Kendall tau distance}.
\begin{definition}[Kendall tau distance]
    Given two permutations $\pi_1,\pi_2\in \sym_d$, the \emph{Kendall tau distance} between them,
    denoted by $\mathcal{K}(\pi_1,\pi_2)$, is the number of pairwise disagreements between $\pi_1$
    and $\pi_2$, i.e.,
    \[
        \mathcal{K}(\pi_1,\pi_2):=|\{ (a,b) \in [d] \times [d] \mid a<_{\pi_1} b \text{ but
        }b<_{\pi_2}a \}|.
    \]
\end{definition}
Another important distance measure is the \emph{Spearman footrule} (aka \emph{Spearman's rho})
which is essentially the $\ell_1$-norm between two permutations.
\begin{definition}[Spearman footrule distance]
    Given two permutations $\pi_1,\pi_2\in \sym_d$, the \emph{Spearman footrule distance} between
    them
    is defined as
    $\mathcal{F}(\pi_1,\pi_2):=\sum_{i\in[d]} |\pi_1(i)-\pi_2(i)|$.
\end{definition}

Another interesting distance measure is the \emph{Ulam distance} which counts the minimum number of
character move operations between two permutations~\cite{AD99}. This definition is motivated by the
classical \emph{edit distance} that is used to measure the dissimilarity between two strings. A
character move operation in a permutation can be thought of as “picking up” a character from its
position and then “inserting” that character into a different position
\footnote{One may also consider
    one deletion and one insertion operation instead of a character move, and define the Ulam
    distance accordingly as in~\cite{CMS01}.}.
\begin{definition}[Ulam distance]
    Given two permutations $\pi_1,\pi_2\in \sym_d$, the \emph{Ulam distance} between them, denoted
    by $\mathcal{U}(\pi_1,\pi_2)$, is the minimum number of character move operations that is needed
    to transform $\pi_1$ into $\pi_2$.
\end{definition}
Alternately, the Ulam distance between $\pi_1,\pi_2$ can be defined as $d - |\mathsf{LCS}(\pi_1,
\pi_2)|$, where $|\mathsf{LCS}(\pi_1,\pi_2)|$ denotes the length of a \emph{longest common subsequence} between the
sequences $\pi_1$ and $\pi_2$.

\paragraph*{Fair rankings.} We are given a set $C$ of $d$ candidates, which are partitioned into
$g$ groups; \AS{and rational vectors $\bar{\alpha}= (\alpha_1,\dots,\alpha_g) \in [0,1]^g\cap\Q^g$, $\bar{\beta}= (\beta_1,\dots,\beta_g)\in [0,1]^g\cap\Q^g$. We call a ranking (of these $d$ candidates) \emph{fair} if a sufficiently large prefix,
has a certain proportion (between $\alpha_i$ and $\beta_i$)} of representatives from each group $i\in[g]$. More formally,

\begin{definition}[$(\bar{\alpha},\bar{\beta})$-$k$-fair ranking]
    \label{def:weak-fair}
    Consider a set $C$ of $d$ candidates partitioned into $g$ groups $G_1,\dots,G_g$, and
    $\bar{\alpha}= (\alpha_1,\dots,\alpha_g) \in [0,1]^g\cap\Q^g$, $\bar{\beta}= (\beta_1,\dots,\beta_g)
    \in [0,1]^g\cap\Q^g$, $k \in [d]$. A ranking $\pi \in \sym_d$ is said to be
    \emph{$(\bar{\alpha},\bar{\beta})$-$k$-fair} if for the $k$-length prefix $P$ of $\pi$ and
    each group $i\in[g]$, there are at least $\lfloor \alpha_i\cdot k \rfloor $ and at most $\lceil
    \beta_i\cdot k \rceil$ elements from the group $G_i$ in $P$, i.e.,
    \[
        \forall_{i \in [g]},\; \lfloor \alpha_i\cdot k \rfloor \le |P\cap G_i| \le \lceil
        \beta_i\cdot k \rceil.
    \]
\end{definition}

\AS{ We also define stronger fairness notions that preserves the proportionate representation for
 multiple prefixes, and not just a fixed $k$-length prefix. For this we first define the notion of
 a block prefix, which is a prefix of some length $b$, such that $b\cdot\alpha_i$ and
 $b\cdot\beta_i$ are integers for all $i\in[g]$.}
\AS{
\begin{definition}[$(\bar{\alpha},\bar{\beta})$-block-$k$-fair ranking]
    \label{def:block-fair}
    \footnote{In a preliminary version of the paper that appeared in NeurIPS'22, we claimed a stronger
     notion of fairness, but some technical claims in it were found to be incorrect due to
     arithmetic involving floor/ceiling. To remedy this, we have now used the slightly weaker notion
     of \AS{block-}fairness here.}
    Consider a set $C$ of $d$ candidates partitioned into $g$ groups $G_1,\dots,G_g$, rational
    vectors
    $\bar{\alpha}= (\alpha_1,\dots,\alpha_g) \in [0,1]^g\cap\Q^g$,
    $\bar{\beta}= (\beta_1,\dots,\beta_g) \in [0,1]^g\cap\Q^g$, $k \in [d]$, and a given
    block-prefix length $b$. A ranking $\pi \in \sym_d$ is said to be
    \emph{$(\bar{\alpha},\bar{\beta})$ block-$k$-fair} if for any prefix $P$ of $\pi$, of length
    at least $k$ and satisfying $|P|\mod{b}\equiv0$, and each group $i\in[g]$, there are at least
    $\alpha_i\cdot|P| $ and at most $\beta_i\cdot|P|$ elements from the group $G_i$ in $P$, i.e.,
    \[
        \forall_{\text{prefix }P: |P| \ge k \land |P|\mod{b}\equiv0},\; \forall_{i \in [g]},\; \alpha_i\cdot|P| \le |P\cap G_i| \le \beta_i\cdot|P|.
    \]
\end{definition}
}

\begin{definition}[$(\bar{\alpha},\bar{\beta})$-strict-$k$-fair ranking]
    \label{def:fair}
    Consider a set $C$ of $d$ candidates partitioned into $g$ groups $G_1,\dots,G_g$, and
    $\bar{\alpha}= (\alpha_1,\dots,\alpha_g) \in [0,1]^g\cap\Q^g$, $\bar{\beta}= (\beta_1,\dots,\beta_g)
    \in [0,1]^g\cap\Q^g$, $k \in [d]$. A ranking $\pi \in \sym_d$ is said to be
    \emph{$(\bar{\alpha},\bar{\beta})$-strict-$k$-fair} if for any prefix $P$ of length at least $k$, of
    $\pi$ and each group $i\in[g]$, there are at least $\lfloor \alpha_i\cdot|P| \rfloor $ and at
    most $\lceil \beta_i\cdot|P| \rceil$ elements from the group $G_i$ in $P$, i.e.,
    \[
        \forall_{\text{prefix }P: |P| \ge k},\; \forall_{i \in [g]},\; \lfloor \alpha_i\cdot|P|
        \rfloor \le |P\cap G_i| \le \lceil \beta_i\cdot|P| \rceil.
    \]
\end{definition}

\AS{For brevity, we will sometimes use the term fair ranking instead
of $(\bar{\alpha},\bar{\beta})$-$k$-fair ranking (and similarly also use the terms block-fair ranking, and strict-fair ranking).}

\section{Closest Fair Ranking}
\label{sec:CFR}
In this section, we consider the problem of computing the closest fair ranking of a given input
ranking. Below we formally define the problem.
\begin{definition}[Closest fair ranking problem]
    Consider a metric space $(\sym_d,\rho)$ for a $d \in \mathbb{N}$. Given a ranking $\pi\in
    \sym_d$ and $\bar{\alpha},\bar{\beta} \in [0,1]^g$ for some $g \in \mathbb{N}$, $k \in [d]$, the
    objective of the \emph{closest fair ranking problem} (resp. closest \AS{block/strict-}fair ranking problem)
    is to find a $(\bar{\alpha},\bar{\beta})$-$k$-fair ranking (resp.
    \AS{$(\bar{\alpha},\bar{\beta})$-block-$k$-fair ranking} or $(\bar{\alpha},\bar{\beta})$-strict-$k$-fair ranking) $\pi^*\in \sym_d$ that minimizes the
    distance $\rho(\pi,\pi^*)$ .
\end{definition}
 The following algorithm assumes that a fair ranking exists, and finds one in this case. If such a
 ranking does not exist, then the procedure might return an arbitrary ranking. However, it is
 possible to check in linear time whether a given ranking is fair or not. Hence, we can correctly
 output a solution if one exists, and return `no solution' otherwise.

\subsection{Closest fair ranking under Kendall tau metric}
\label{sec:KT-CFR}

\paragraph*{Closest fair ranking.}We first show that we can compute a closest fair
ranking under the Kendall tau metric exactly in linear time\AS{, and later generalize our result
to the notion of block-fairness.}
\begin{restatable}{theorem}{AltGreedy} \label{thm:AltGreedy}
    There exists a linear time algorithm that, given a ranking $\pi\in \sym_d$, a partition of $[d]$
    into $g$ groups $G_1,\dots,G_g$ for some $g \in \mathbb{N}$, and
    $\bar{\alpha}=(\alpha_1,\dots,\alpha_g) \in [0,1]^g$, $\bar{\beta}=(\beta_1,\dots,\beta_g) \in
    [0,1]^g$, $k \in [d]$, outputs a closest $(\bar{\alpha},\bar{\beta})$-$k$-fair ranking
    under the Kendall tau distance.
\end{restatable}

Let us first describe the algorithm. Our algorithm follows a simple greedy strategy. For each group
$G_i$, it picks the top $\lfloor \alpha_i k \rfloor$ elements according to the input ranking $\pi$,
and add them in a set $P$. If $P$ contains $k$ elements, then we are done. Otherwise, we iterate
over the remaining elements (in the increasing order of rank by $\pi$) and add them in $P$ as long
as for each group $G_i$, $|P \cap G_i| \le
\lceil \beta_i k \rceil$ (each group has at most $\lceil \beta_i k \rceil$ elements in $P$) until
 the size of $P$ becomes exactly $k$. Then we use the relative ordering of the elements in $P$ as
 in the input ranking $\pi$ and make it the $k$-length prefix of the output ranking $\sigma$. Fill
 the last $d-k$ positions of $\sigma$ by the remaining elements ($[d] \setminus P$) by following
 their relative ordering as in the input $\pi$. \AS{We now give a more formal pseudocode for the algorithm:}

\begin{algorithm}
    \caption{Algorithm to compute closest fair ranking under Kendall tau}
    \label{alg:AltGreedy}
    \KwIn{Input ranking $\pi \in \sym_d$, $g$ groups $G_1,\dots,G_g$,
    $\bar{\alpha}=(\alpha_1,\dots,\alpha_g) \in [0,1]^g$, $\bar{\beta}=(\beta_1,\dots,\beta_g) \in
        [0,1]^g$, $k \in [d]$.}
    \KwOut{An $(\bar{\alpha},\bar{\beta})$-$k$-fair ranking.}
    Initialize a set $P \leftarrow \emptyset$.\\
    Initialize a ranking $\sigma \leftarrow \emptyset$. \\
    From each group $i \in [g]$, pick the top $\lfloor \alpha_i k \rfloor $ elements of $G_i$
        according to the ranking $\pi$ and add them to $P$.\\
    \If{$|P| > k$}{
        \Return{No fair ranking exists.}
    }
    \Else{
        Iterate over the remaining elements (in increasing order of rank by $\pi$) and add them in
        $P$ as long as $|P| \le k$ and for each $i \in [g]$, $|P \cap G_i| \le \lceil \beta_i k
        \rceil$.\\
        $\sigma\leftarrow$ Construct a new ranking by first placing the elements of $P$ in the top $k$
        positions (i,e., in the $k$-length prefix) while preserving the relative ordering among
        themselves according to $\pi$, and then placing the remaining elements ($[d] \setminus P$)
        in the rest (the last $d - k$ positions) while preserving the relative ordering among
        themselves according to $\pi$.\\
    }

    Iterate over ranking $\sigma$, and count the fraction of elements in the top-$k$ from each group. \\
    \If{$\sigma$ is $(\bar{\alpha},\bar{\beta})$-$k$-fair} {
        \Return{$\sigma$.}
    }
    \Else{
        \Return{No fair ranking exists.}
    }

\end{algorithm}

In the algorithm description above, for a subset $P\subseteq[d]$, we say that an ordering $\sigma$
of the elements of $P$ \textit{preserves relative orderings according to $\pi$}, when for any two
elements $a\neq b\in P$, $a<_\sigma b$ if and only if $a<_\pi b$.

By the construction of set $P$, at the end, for each group $G_i$, $ \lfloor \alpha_i k \rfloor
\le |P \cap G_i| \le \lceil \beta_i k \rceil$. Since we use the elements of $P$ in the $k$-length
prefix of the output ranking $\sigma$, $\sigma$ is an $(\bar{\alpha},\bar{\beta})$-$k$-fair
ranking. For the running time, a straightforward implementation our algorithms takes $\bigO (d)$
time.
It only remains to argue that $\sigma$ is a closest $(\bar{\alpha},\bar{\beta})$-$k$-fair
ranking to the input $\pi$. To show that, we use the following key observation.

\begin{restatable}{claim}{RelOrderProp} \label{clm:RelOrderProp}
    Under the Kendall tau distance, there always exists a closest $(\bar{\alpha},\bar{\beta})$-
    $k$-fair ranking $\pi^*$ such that, for each group $G_i$ ($i \in [g]$), for any two elements $a
    \ne b \in G_i$, $a <_{\pi^*} b$ if and only if $a <_{\pi} b$.
\end{restatable}
\begin{proof}
    Let $\pi$ be the input ranking, and $\pi'$ be some closest $(\bar{\alpha},\bar{\beta})$-$k$-fair ranking.
    If there exists $x,y\in G$ for some group $G$, such that $x<_\pi y$ but $x>_{\pi'}y$, then we claim
    that the permutation $\pi''$ obtained by swapping $x$ and $y$, is fair, and satisfies $\mathcal{K}
    (\pi,\pi'')\le\mathcal{K}(\pi,\pi')$. Hence, by a series of such swap operations, any $(\bar
    {\alpha},\bar{\beta})$-$k$-fair ranking $\pi'$, can be transformed into one that satisfies the property
    that, $a\ne b \in G_i$, $a <_{\pi'} b$ if and only if $a <_{\pi} b$.

    We partition the set $[d]\setminus\{x,y\} $ into three sets, $L=\{z:z<_{\pi''}x\},
    B=\{z:x<_{\pi''}z<_{\pi''}y\}, U=\{z:y<_{\pi''}z\}$. Now, consider the sets
    $K_{\pi'}=\{ (i,j)\in[d]\times[d]\mid (i<_{\pi}j)\land (i>_{\pi'}j) \},
    K_{\pi''}=\{ (i,j)\in[d]\times[d]\mid (i<_{\pi}j)\land (i>_{\pi''}j) \}$.

    Consider any element $(a,b)\in K_{\pi''}$, such that $a, b$ belong to the same group.
    \begin{itemize}
        \item If $a,b\not\in\{x,y\}$: then clearly the pair of elements is in the same positions in
            $\pi'$ also, and hence $(a,b)\in K_{\pi'}$.
        \item If $y\in\{a,b\}$: let $a'\in\{a,b\}\setminus\{y\}$.
            \begin{itemize}
                \item $a'\in U$: This means that $x<_{\pi''}y<_{\pi''}a'$, so with swapped $x$ and
                    $y$, we will have that $y<_{\pi'}x<_{\pi'}a'$. The relative position of $y$ and
                    $a'$ remains the same and hence $(a',y)\in K_{\pi'}$.
                \item $a'\in B$: This means that $x<_{\pi''}a'<_{\pi''}y$, so with swapped $x$ and
                    $y$, we will have that $y<_{\pi'}a'<_{\pi'}x$. But we know that $(y,a')\in K_{\pi''}$, which means that $x<_{\pi}y<_{\pi}a'$. So we have that $x<_{\pi}a'\land
                    a'<_{\pi'}x$, which gives us that $(x,a')\in K_{\pi'}$.
                \item $a'\in L$: This means that $a'<_{\pi''}x<_{\pi''}y$, so with swapped $x$ and
                    $y$, we will have that $a'<_{\pi'}y<_{\pi'}x$. The relative position of $y$ and
                    $a'$ remains the same and hence $(y,a')\in K_{\pi'}$.
            \end{itemize}
        \item If $x\in\{a,b\}$: This case can also be analyzed similarly to the case above.
    \end{itemize}
    Now, note that $(x,y)$ is not in $K_{\pi''}$, so we get that $|K_{\pi''}|\le |K_{\pi'}|$.
    Also, since this swap operation does not change the number (or fraction) of elements from any group in the top-$k$ ranks, $\pi''$ is still fair. This completes our proof.
\end{proof}

We say that a ranking $\pi^*$ satisfying the above property (i.e., for each group $G_i$, $i \in [g]$, for any two elements $a\ne b \in G_i$, $a <_{\pi^*} b$ if and only if $a <_{\pi} b$), \textit{preserves intra-group orderings}. This is because for elements of any group, their ordering in $\pi^*$ is the same as their ordering in $\pi$.

We want to highlight that the above claim holds for both the notions of fairness we consider. \AS{Now, using the
above claim, we complete the proof of} \autoref{thm:AltGreedy}.
\begin{proof}[Proof of \autoref{thm:AltGreedy}.]
    Let $\pi$ be the given input ranking, $\pi_\opt$ be a closest ranking that also preserves
    intra-group orderings (the existence of which is guaranteed by~\autoref{clm:RelOrderProp}), and $\pi_\GRD$ be the ranking returned by Algorithm~\ref{alg:AltGreedy}. First, observe, by the construction of $\pi_\GRD$ (in Algorithm~\ref{alg:AltGreedy}), $\pi_\GRD$ also preserves intra-group orderings, and moreover, the relative ordering among the set of symbols in $\pi_\GRD[k]$ (and in $[d]\setminus \pi_\GRD[k]$) is the same as that in $\pi$. It is also easy to see that the relative ordering among the set of symbols in $\pi_\OPT[k]$ (and in $[d]\setminus \pi_\OPT[k]$) is the same as that in $\pi$; otherwise, by only changing the relative ordering of the set of symbols in $\pi_\OPT[k]$ (and in $[d]\setminus \pi_\OPT[k]$) to that in $\pi$, we get another fair ranking $\pi'$ such that $\mathcal{K}(\pi',\pi) < \mathcal{K}(\pi_\OPT,\pi)$, which leads to a contradiction.

    Now, assume towards contradiction that the rankings $\pi_{\opt}$ and $\pi_\GRD$ are not the same, which implies (from the discussion in the last paragraph) that there is a symbol $a \in \pi_\GRD[k] \setminus \pi_\OPT[k]$ and $b \in \pi_\OPT[k] \setminus \pi_\GRD[k]$. Since both $\pi_\OPT$ and $\pi_\GRD$ preserve intra-group orderings, $a$ and $b$ cannot belong to the same group, i.e., $a\in G_i$ and $b \in G_j$ such that $i \ne j$. Then since $\pi_\OPT$ is a fair ranking that preserves intra-group orderings,
    \begin{enumerate}
        \item\label{itm:a} element $a$ is not among the top $\lfloor \alpha_i k \rfloor$ elements of $G_i$ according to $\pi$; and
        \item\label{itm:b} element $b$ is among the top $\lceil \beta_j k \rceil$ elements of $G_j$ according to $\pi$.
    \end{enumerate}

        Next, we argue that $ a <_\pi b$. If not, since $ a <_{\pi_\GRD} b$, it follows from the construction of $\pi_\GRD$ (in Algorithm~\ref{alg:AltGreedy}), either $a$ is among the top $\lfloor \alpha_i k \rfloor$ elements of $G_i$ according to $\pi$ contradicting Item~\ref{itm:a}, or $b$ is not among the top $\lceil \beta_j k \rceil$ elements of $G_j$ according to $\pi$ contradicting Item~\ref{itm:b}. Hence, we conclude that $ a <_\pi b$.

        We now claim that by swapping $a$ and $b$ in $\pi_\opt$, we get a ranking $\pi'$
    that is fair and also reduces the Kendall tau distance from $\pi$, giving us a contradiction. Let us start by arguing that $\pi'$ is fair. Note, $\pi'[k] = (\pi_\OPT[k] \setminus \{b\}) \cup \{a\}$. Next, observe, since $\pi_\GRD$ is a fair ranking that preserves intra-group orderings,
    \begin{itemize}
        \item element $b$ is not among the top $\lfloor \alpha_j k \rfloor$ elements of $G_j$ according to $\pi$; and
        \item element $a$ is among the top $\lceil \beta_i k \rceil$ elements of $G_i$ according to $\pi$.
    \end{itemize}
    As a consequence, swapping $a$ and $b$ in $\pi_\OPT$ to get $\pi'$ does not violate any of the fairness constraints for the groups $G_i$ and $G_j$ (and, of course, none of the other groups). Hence, $\pi'$ is also a fair ranking. So it remains to argue that $\mathcal{K}(\pi',\pi) < \mathcal{K}(\pi_\OPT,\pi)$ to derive the contradiction.

    Observe that the set of pair of symbols for which the relative ordering changes in $\pi'$ as compared to
    $\pi_\opt$ is the following:
    \[
    \{(b,c) , (c,a) \mid b <_{\pi_\OPT} c <_{\pi_\OPT} a\}.
    \]

    Consider any element $c$ such that $b <_{\pi_\OPT} c <_{\pi_\OPT} a$. Note, $a <_{\pi'} c <_{\pi'} b$.
    \begin{itemize}
        \item If the pair $(c,a)$ creates an inversion (with respect to $\pi$) in $\pi'$, but not in $\pi_\OPT$, then $c<_\pi a$. Now, since $a<_\pi b$, we have $c<_\pi b$. Thus the pair $(b,c)$ creates an inversion (with respect to $\pi$) in $\pi_\OPT$, but not in $\pi'$.

        \item If the pair $(b,c)$ creates an inversion (with respect to $\pi$) in $\pi'$, but not in $\pi_\OPT$, then $b<_\pi c$. Now, since $a<_\pi b$, we have $a<_\pi c$. Thus the pair $(c,a)$ creates an inversion (with respect to $\pi$) in $\pi_\OPT$, but not in $\pi'$.
    \end{itemize}
    Hence, we see that $\pi'$ has at least one lesser inverted pair (being $a$ and $b$) than
    $\pi_\opt$. In other words, $\mathcal{K}(\pi',\pi) < \mathcal{K}(\pi_\OPT,\pi)$, leading to a contradiction, which completes the proof.
\end{proof}

\paragraph*{Extension to \AS{block-}fairness notion.}Previously, we provide an algorithm that outputs a
fair ranking (see Definition~\ref{def:weak-fair} closest
to the input. Now, we present an algorithm that outputs a closest \AS{block-}fair ranking (according to Definition~\ref{def:block-fair}).

\begin{restatable}{theorem}{StrongGreedy} \label{thm:StrongGreedy}
    There exists an $\bigO(d^2)$ time algorithm that, given a ranking $\pi\in \sym_d$, a partition of
    $[d]$ into $g$ groups $G_1,\dots,G_g$ for some $g \in \mathbb{N}$, and
    $\bar{\alpha}=(\alpha_1,\dots,\alpha_g) \in [0,1]^g\cap\Q^q$, $\bar{\beta}=(\beta_1,\dots,\beta_g) \in
    [0,1]^g\cap\Q^q$, $k \in [d]$, outputs a closest $(\bar{\alpha},\bar{\beta})$-\AS{block k-}fair ranking
    under the Kendall tau distance.
\end{restatable}

The main challenge with this stronger fairness notion is that now we need to satisfy the fairness
criteria for all the \AS{block}-prefixes not just a fixed $k$-length prefix as in the
\AS{previous case.}
Surprisingly, we show that under the Kendall tau metric, by iteratively applying the algorithm for
the closest fair ranking (Algorithm~\ref{alg:AltGreedy}) as a black-box, over the \AS{block-}prefixes of
decreasing length, we can construct a closest \AS{block}-fair ranking (not just a closest fair ranking).
Here, it is worth noting that at any iteration the input to Algorithm~\ref{alg:AltGreedy} is a
prefix of $\pi$ which may not be a permutation. However, Algorithm~\ref{alg:AltGreedy} only treats the
input as a sequence of numbers (not really as a permutation). \AS{We now give} a formal description of the algorithm.

\begin{algorithm}
    \caption{Algorithm to compute closest \AS{block-k-}fair ranking under Kendall tau}
    \label{alg:StrongGreedy}
    \KwIn{Input ranking $\pi \in \sym_d$, $g$ groups $G_1,\dots,G_g$,
    $\bar{\alpha}=(\alpha_1,\dots,\alpha_g) \in [0,1]^g\cap\Q^g$, $\bar{\beta}=(\beta_1,\dots,\beta_g) \in
    [0,1]^g\cap\Q^g$, $k \in [d]$, $b\in\Z_+$.}
    \KwOut{An $(\bar{\alpha},\bar{\beta})$ \AS{block}-k-fair ranking.}

    \AS{$p\leftarrow\left\lfloor\frac{d}{b}\right\rfloor\times b$}. \\
    \AS{\While{$p\ge k$} {
        $\pi \leftarrow$ Algorithm~\ref{alg:AltGreedy}
        ($\pi,(G_1,\dots,G_g),\bar{\alpha},\bar{\beta},p)$. \\
        $p\leftarrow p-b$.
    }}
    \Return{The resulting ranking stored in variable $\pi$.}
\end{algorithm}

Note that, since we iteratively apply Algorithm~\ref{alg:AltGreedy} on a prefix of
$\pi$ (not the whole sequence represented by $\pi$), it is not even clear whether the algorithm
finally outputs a fair ranking (assuming it exists). Below we first argue that if there exists a
\AS{block}-fair ranking, then the output $\sigma$ must be such a ranking. Next, we establish that
$\sigma$ is indeed a closest \AS{block}-fair ranking to $\pi$.

Let $\pi^*$ be a closest \AS{block}-fair ranking to $\pi$ that preserves intra-group orderings
(Claim \ref{clm:RelOrderProp} guarantees its existence). We
show that the output $\sigma = \pi^*$. We start the argument by considering any two \AS{block-}prefixes of
length $k_1$ and $k_2$, where $k_2 < k_1$. We argue that $k_1$ and $k_2$-length prefixes of
$\sigma$ and $\pi^*$ are the same. Since this holds for any such $k_1$ and $k_2$,
where $k_2 < k_1$, by using induction we can show that $\sigma = \pi^*$.

For the sake of analysis, let us consider the following three permutations. Let $\pi_1$ be the
$(\bar{\alpha},\bar{\beta})$-$k_1$-fair ranking closest to $\pi$, output by
Algorithm~\ref{alg:AltGreedy}. Let $\pi_2$ be the ranking output by Algorithm~\ref{alg:AltGreedy}
when given the $k_1$-length prefix of $\pi_1$ (i.e., the sequence $\pi_1([k_1])$) as input and is
asked to output an $(\bar{\alpha},\bar{\beta})$-$k_2$-fair ranking closest to $\pi_1$. Further,
let $\pi'_2$ be the $(\bar{\alpha},\bar{\beta})$-$k_2$-fair ranking closest to $\pi$, output by
Algorithm~\ref{alg:AltGreedy}. In other words, $\pi_2$ is the ranking produced by first applying
Algorithm~\ref{alg:AltGreedy} on $\pi$ to make its $k_1$-length prefix fair and then apply
Algorithm~\ref{alg:AltGreedy} again on that output to make its $k_2$-length prefix fair. Whereas,
$\pi'_2$ is the ranking produced by directly applying Algorithm~\ref{alg:AltGreedy} on $\pi$ to
make its $k_2$-length prefix fair.

From the definition it is not at all clear whether such a $\pi_2$ even exists. The next claim argues
about the existence of ranking $\pi_2$.
\begin{restatable}{claim}{StrongExist}
    If there is a ranking $\pi'$ such that its $k_1$-length prefix $P_1$ and $k_2$-length prefix
    $P_2$ satisfies that for each group $G_i$ ($i \in [g]$), $\alpha_i k_1 \le |P_1
    \cap G_i| \le \beta_i k_1 $ and $\alpha_i k_2 \le |P_2 \cap G_i|
    \le \beta_i k_2 $, then $\pi_2$ exists.
\end{restatable}
\begin{proof}
    Our proof is constructive. Since $\pi'$ exists, $\pi_1$ also exists (follows from the
    correctness of Algorithm~\ref{alg:AltGreedy}). Since for any group $G_i$, $|\pi_1([k_1]) \cap
    G_i| \ge \alpha_i k_1 \ge \alpha_i k_2 $, the elements in the
    prefix $\pi_1([k_1])$ are sufficient to satisfy the lower bound fairness constraints for the
    $k_2$-length prefix.

    For a group $i\in[g]$, let $\ell_i$ be the number of items of group $G_i$ in the prefix
    $\pi_1([k_1])$. Next, we argue that assuming the existence of $\pi'$ (as in the claim
    statement), it is always possible to construct a $\pi_2$ for which the upper bound fairness
    constraints are satisfied for the $k_2$-length prefix. Note, $\sum_{i \in [g]} \ell_i = k_1$ and
    for each $i \in [g]$, $\ell_i \le \beta_i k_1$.

    Consider $L_i'$ to be the set of top $ \ell_i \times\frac{k_2}{k_1}$
    elements from $\pi_1([k_1]) \cap G_i$. Let $|L'_i| = \ell'_i$. Then,
    \[\sum_{i\in[g]}\ell_i'=\sum_{i\in[g]} \ell_i\times\frac{k_2}{k_1}\ge
    \frac{k_2}{k_1}\times\sum_{i\in[g]}\ell_i \ge k_2.\]
    Also note that for each $i \in [g]$,
    \[\ell_i'= \ell_i\times\frac{k_2}{k_1}
        \le\beta_i k_1\times\frac{k_2}{k_1}
    \le\beta_i k_2\] So, by ranking the above set $L'_i$ of chosen elements,
     we can ensure that we can satisfy both the upper and lower bound fairness constraints for the
     $k_2$-prefix. Thus we get an $(\bar{\alpha},\bar{\beta})$-$k_2$-fair ranking whose
     elements in the $k_2$-prefix are from the set $\pi_1([k_1])$. Since one such valid solution
     exists, Algorithm~\ref{alg:AltGreedy} can find an appropriate valid solution $\pi_2$.
\end{proof}

It can further be shown that,
\begin{restatable}{claim}{clma} \label{clm:a}
    The set of elements in $\pi_2([k_1])$ is the same as that in $\pi_1([k_1])$.
\end{restatable}
\begin{proof}
    We only call the Algorithm~\ref{alg:AltGreedy} subroutine on the $k_1$-length prefix of $\pi_1$
    to construct $\pi_2$. So by construction, we have that for the $k_1$-length prefix, the set of
    elements in both the permutations are the same.
\end{proof}

\begin{restatable}{claim}{clmb} \label{clm:b}
    The set of elements in $\pi_2([k_2])$ is the same as that in $\pi'_2([k_2])$.
\end{restatable}
\begin{proof}
    Consider an element $e\in \pi_2'([k_2]) \cap G_i$ for some $i \in [g]$. If $e$ is among the top
    $\alpha_i k_2$ elements (according to $\pi$) inside the group $G_i$, then by
    Algorithm~\ref{alg:AltGreedy}, it would also be selected in $\pi_1([k_1])$ (since $k_1 \ge k_2$)
    and also in $\pi_2([k_2])$.

    \AS{Now consider the case where $e$ is among the top $\beta_i k_2$ elements of $G_i$,
    but not among the top $\alpha_i k_2 $ elements. This means that $e$ is also
    among the top $\beta_i k_1 (\ge \beta_i k_2)$ elements of its group
    $G_i$. So, if it is encountered during the execution of
    Algorithm~\ref{alg:AltGreedy} on $\pi$ to get an $(\bar{\alpha},\bar{\beta})$-$k_1$-fair
    ranking, then it will be selected in $\pi_1([k_1])$ (and similarly for $\pi_2([k_2])$).
    So to complete our proof we need to show that element $e$ will still be encountered while
    constructing $\pi_2$. We do this in two steps:}

    \AS{\textbf{Step 1, To show $e\in\pi_1([k_1])$:} We know that $e$ was encountered (and picked) during the execution of Algorithm~\ref{alg:AltGreedy} on $\pi$ for a prefix of size $k_2$. Let the elements that were encountered before $e$, but not selected due to violating upper bound constraints, be from the subset $I'\subseteq[g]$ of groups. Since some elements from these groups were not selected, we know that $k_2\cdot\sum_{i\in I}\alpha_i + k_2\cdot\sum_{i\in I'}\beta_i < k_2$ which gives is that $\sum_{i\in I}\alpha_i+\sum_{i\in I'}\beta_i < 1$. Now, when we execute Algorithm~\ref{alg:AltGreedy} on $\pi$ for a prefix of size $k_1$, we know that $k_1\cdot\sum_{i\in I}\alpha_i + k_1\cdot\sum_{i\in I'}\beta_i < k_1$. Hence element $e$ will be encountered and selected in $\pi_1([k_1])$.}

    \AS{\textbf{Step 2, To show $e\in\pi_2([k_2])$:} By construction $\pi_1([k_1])$ is such that $x<_{\pi_1}e$ which also means $x<_\pi e$. So if $e$ was encountered (and picked) during execution of execution of Algorithm~\ref{alg:AltGreedy} on $\pi$ for a prefix of size $k_2$, then it will also be encountered (and picked) during execution of Algorithm~\ref{alg:AltGreedy} on $\pi_1$ for a prefix of size $k_2$.
    }

    Since the sizes of both the sets (of top $k_2$ ranks of $\pi_2$ and $\pi_2'$) are equal, the
    two sets are in fact the same, and so are the rankings (by Algorithm~\ref{alg:AltGreedy}).
\end{proof}

\begin{restatable}{claim}{clmd} \label{clm:d}
    The set of elements in $\pi^*([k_1])$ is the same as that in $\pi_2([k_1])$.
\end{restatable}
\begin{proof}
    Assume towards contradiction that there exists $a\in\pi_2([k_1])\setminus\pi^*([k_1])$
    and $b\in\pi^*([k_1])\setminus\pi_2([k_1])$. If $a, b$ were in the same group, then by
    Algorithm~\ref{alg:AltGreedy}, we know that $a<_{\pi}b$, and hence by swapping the elements in
    $\pi^*$, the distance from $\pi$ can only be reduced (as shown in latter part of proof
    of \autoref{thm:AltGreedy}). Hence we can obtain a different solution $\bar{\pi}$ in which
    $a\in \bar{\pi}([k_1])$ and $b \not \in \bar{\pi}([k_1])$, and this is also \AS{block}-fair.
    This contradicts that $\pi^*$ is a closest \AS{block}-fair ranking to $\pi$ that preserves intra-group orderings.

    In the other case, $a$ and $b$ are not in the same group, i.e., $a\in G_i$ and $b\in G_j$ for some $i \ne j$.
    Now we note that $a$ cannot be among the top $\alpha_i k_1 $ elements, but is in
    the top $\beta_i k_1 $ elements in $G_i$. Similarly, $b$ cannot be among the top
    $\alpha_j k_1 $ elements, but is in the top $\beta_j k_1 $ elements
    in $G_j$. Again, it follows from \autoref{thm:AltGreedy}, $a<_{\pi}b$, and that by swapping
    these two elements in $\pi^*$ we can only reduce the distance from $\pi$, while obtaining
    another \AS{block}-fair ranking. This again
    contradicts that $\pi^*$ is a closest \AS{block}-fair ranking to $\pi$ that preserves intra-group orderings.
    The claim now follows.
\end{proof}

\begin{restatable}{claim}{clmc} \label{clm:c}
    The set of elements in $\pi^*([k_2])$ is the same as that in $\pi_2([k_2])$.
\end{restatable}
\begin{proof}
    From Claim~\ref{clm:b} we know that \(\pi_2([k_2])=\pi_2'([k_2])\). So, it suffices to prove
    that the sets $\pi^*([k_2])$ and $\pi_2'([k_2])$ are equal. Note that this amounts to showing
    that for some prefix, the output of Algorithm~\ref{alg:AltGreedy} and the optimal solution
    have the same set of elements. The proof is hence very similar to that of Claim~\ref{clm:d}.
\end{proof}

We now apply Claim~\ref{clm:d} and Claim~\ref{clm:c}
iteratively to complete the correctness of Algorithm~\ref{alg:StrongGreedy}.

\begin{proof}[Proof of \autoref{thm:StrongGreedy}.]
    We show by induction that Algorithm~\ref{alg:StrongGreedy} in fact outputs the same ranking
    (referred to as the greedy solution) as the optimal \AS{block}-fair ranking $\pi^*$, which preserves
    intra-group orderings. In the induction we consider a \AS{block-}prefix, at the step at
    which subroutine Algorithm~\ref{alg:AltGreedy} was executed on it.
    \begin{description}
        \item[Hypothesis:] After Algorithm~\ref{alg:StrongGreedy} calls the subroutine on a
            \AS{block-}prefix
            of length $d-i$, for some $i<d$, the result is the optimal $(\bar{\alpha},\bar{\beta})$
            \AS{block}-$(d-i)$-fair ranking, which preserves intra-group orderings.
        \AS{\item[Base Case:] For $i=d\mod{b}$, the prefix of length $d-i=\lfloor d/b\rfloor$
            is the largest block-prefix considered by the algorithm. On application of subroutine
            Algorithm~\ref{alg:AltGreedy}, we know that we get a $(\bar{\alpha},\bar{\beta})$-block-$(d-i)$-fair ranking, which preserves intra-group orderings.}
        \item[Induction Step:] Let $P_2$ be the prefix of length $d-(i+\AS{b})$ at which the
            subroutine Algorithm~\ref{alg:AltGreedy} was just executed. So the subroutine was executed
            on on a prefix $P_1$ of length $d-i$ in the previous step, and hence it is already
            \AS{block-}fair (by induction hypothesis).

            From Claim~\ref{clm:c} and Claim~\ref{clm:d} we have that both the greedy
            (Algorithm~\ref{alg:StrongGreedy}'s output) and the optimal solution have the same set
            of elements in both the top $P_1$ and $P_2$ prefixes.  And we know that the optimal
            solution (by definition) and the greedy solution (by construction) preserve relative
            orderings w.r.t. $\pi$. This implies that the greedy solution $\pi_2$ is in fact
            the same as the optimal solution $\pi^*$.
    \end{description}
    \AS{Since Algorithm~\ref{alg:StrongGreedy} stops at the appropriate block-prefix of size
    greater than $k$, the output
    is in fact an $(\bar{\alpha},\bar{\beta})$-block-$k$-fair ranking.}
\end{proof}

It only remains to argue that the algorithm runs in time $\bigO(d^2)$. This is easy to see since
the algorithm invokes at most $d$ calls to the $\bigO(d)$ subroutine Algorithm~\ref{alg:AltGreedy}.

\subsection{Closest fair ranking under Ulam metric}
The \AS{following result} uses an intricate dynamic program exploiting the connection between the Ulam
distance with the Longest Common Subsequence problem.
\label{sec:Ulam-CFR}
\begin{restatable}{theorem}{UlamDP} \label{thm:ulamdp}
    There exists a polynomial time dynamic programming based algorithm that finds a
    $(\bar{\alpha},\bar{\beta})$\AS{-strict}-$k$-fair ranking under Ulam metric when there are constant
    number of groups.
\end{restatable}

\begin{algorithm} \label{alg:UlamDP}
    \caption{DP algorithm for $(\bar{\alpha},\bar{\beta})$\AS{-strict}-k-fair ranking under Ulam}
    \KwIn{Input ranking $\pi$, $g$ groups $G_1, \dots, G_g$, vectors
        $\bar{\alpha}=(\alpha_1,\dots,\alpha_g)\in[0,1]^g,\bar{\beta}=(\beta_1,\dots,\beta_g)\in[0,1]^g$,
    $k\in[d]$.}
    \KwOut{\AS{Filled DP table with correct LCS values and strings.}
    }
    \SetKwFunction{FMain}{FnDP}
    \SetKwProg{Fn}{Function}{:}{}
    DP$[d][|G_1|]\dots[|G_g|] := \vec{0}$. \\
    LCS$[d][|G_1|]\dots[|G_g|] := \emptyset$. \\
    \For(\tcc*[f]{Base cases}){$j\in[d]$} {
        \For{$i\in[g]$} {
            DP$[j][0]\dots[a_i=1]\dots[0]=1$ if there is an element of $G_i$ in $\pi[1\dots j]$. \\
            LCS$[j][0]\dots[a_i=1]\dots[0]=$ first element of group $G_i$.
        }
    }
    \Fn{\FMain{$[x][y_1]\dots[y_g]$}} {
        \If{ DP$[x][y_1\dots[y_g] \neq 0$} {
            \Return{DP$[x][y_1]\dots[y_g]$}.
        }
        $\gamma:=\sum_{z\in[g]} y_z$. \\
        \For(\tcc*[f]{DP recurrence loop}){$z\in[g]$} {
            \If{($\gamma\ge k$) AND \big($\lfloor\alpha_z\gamma\rfloor> y_z$ OR $y_z>\lceil\beta_z\gamma\rceil$ OR $y_z>|G_z|$\big) } {
                \Return{$-\infty$}.\tcc*[f]{Fairness constraints}
            }
            \If{ $\pi[x]$ is in $G_z$ } {
                \If{\FMain$[x-1][y_1]\dots[y_z-1]\dots[y_g]+1>$DP$[x][y_1]\dots[y_g]$}{
                    DP$[x][y_1]\dots[y_g] = \FMain[x-1][y_1]\dots[y_z-1]\dots[y_g]+1$.\\
                    To obtain LCS$[x][y_1]\dots[y_g]$, append $\pi[x]$ to LCS$[x-1][y_1]\dots[y_z-1]$.
                }
            }
            \Else{
                \If{\FMain$[x-1][y_1]\dots[y_z-1]\dots[y_g]>$DP$[x][y_1]\dots[y_g]$}{
                    DP$[x][y_1]\dots[y_g] = \FMain[x-1][y_1]\dots[y_z-1]\dots[y_g]$.\\
                    LCS$[x][y_1]\dots[y_g] = \text{LCS}[x-1][y_1]\dots[y_z-1]\dots[y_g]$.
                }
            }
            \If{there are more than $y_z-1$ elements of $G_z$ in $\pi[1\dots j]$}{
                \If{\FMain$[x][y_1]\dots[y_z-1]\dots[y_g]+1>$DP$[x][y_1]\dots[y_g]$}{
                    DP$[x][y_1]\dots[y_g]=\FMain[x][y_1]\dots[y_z-1]\dots[y_g] + 1$.\\
                    Choose an element $g\in\{\pi[0\dots x]\}\setminus\{\sigma[x][y_1]\dots[y_z-1]\dots[y_g]\}$.\\
                    Identify the first element $a$ (in increasing order of $\pi$) in the LCS, satisfying $g<_\pi a$
                    (note that $a$ could be null, representing the last position in the LCS).\\
                    To obtain LCS$[x][y_1]\dots[y_g]$: insert $g$ preceding $a$ in LCS$[x][y_1]\dots[y_z-1]\dots[y_g]$
                    (if $a$ is null, append $g$ to LCS$[x][y_1]\dots[y_z-1]\dots[y_g]$).
                }
            }
            \Else{
                \If{\FMain$[x][y_1]\dots[y_z-1]\dots[y_g]>$DP$[x][y_1]\dots[y_g]$}{
                    DP$[x][y_1]\dots[y_g]=\FMain[x][y_1]\dots[y_z-1]\dots[y_g]$.\\
                    LCS$[x][y_1]\dots[y_g]=\text{LCS}[x][y_1]\dots[y_z-1]\dots[y_g]$.
                }
            }
        }
    }
    DP$[d][|G_1|]\dots[|G_g|] = \FMain[d][|G_1|]\dots[|G_g|]$. \\
\end{algorithm}

\begin{proof}
    Let $\pi$ be the given input string and $g$ be the number of groups. Let $a_i\in
    \mathbb{N}$, for all $i \in [g]$, and $\gamma:=\sum_{i\in[g]}a_i$. Let $\mathcal{P}
     (a_1,\dots,a_g)$ be the family of strings of length $\gamma$ that have exactly $a_i$ elements
     from group $G_i$ \AS{(with no repetitions)} for all $i\in[g]$. Let $\sigma_j
     (a_1,\dots,a_g)$ be the string in this family that has the longest common subsequence
     (LCS) with $\pi[1 \dots j]$ (So the string $\sigma_d(|G_1|,\dots,|G_g|)$ is the string of
     length $d$, that has maximum LCS with $\pi$, and from the alternate definition of the Ulam
     metric, the smallest Ulam distance from $\pi$).

    \AS{We first define a dynamic program to compute the lengths of, and build up, the LCS strings.
     Once we have the required LCS string, we backtrack over the DP to construct the corresponding
     \AS{-strict}fair ranking from it.} The DP subproblem is defined as follows: DP$[j][a_1]\dots[a_g]$ will
     store the length of the longest common subsequence between $\pi[1\dots j]$ and $\sigma_j
     (a_1,\dots,a_g)$. W.l.o.g., we assume that it stores the \AS{LCS string (denoted by LCS$_j
     (a_1,\dots,a_g)$)} as well. With this definition, we  construct a solution for the subproblem
     DP$[j][a_1]\dots[a_g]$ from `smaller' subproblems as follows:

    {\bf Case 1: From subproblems corresponding to smaller values of $\gamma$.} \\
    {\bf Case 1A.} Subproblems  due to $\pi[1 \dots j]$: If for  $i\in[g]$, the number of elements
    from $G_i$ in $\pi[1\dots j]$ is greater than $a_i-1$ then we know that $\pi[1\dots j]$ has an
    LCS of length DP$[j][a_1]\dots[a_i-1]\dots[a_g]$ with $\sigma_j(a_1,\dots, a_i-1,\dots, a_g)$,
    which does not contain all of its elements from $G_i$. \AS{So if we pick one such element $p\in
    G_i$, we can identify an appropriate position such that it gets included in the LCS. Let us say
    we identify the first element $q$ in the LCS that follows $p$ in $\pi$ (i.e., first element $q$
    from the LCS, satisfying $p<_\pi q$). Then if we place $p$ right before $q$ in LCS$_j
    (a_1,\dots, a_i-1,\dots, a_g)$ (and also in $\sigma_j(a_1,\dots,a_i-1,\dots, a_g)$) then we see
    that the size of the LCS has been increased by one.}

    Otherwise, if for $h\in [g]$, $\pi[1\dots j]$ has at most $a_h-1$ elements from $G_h$, then DP$
    [j][a_1]\dots[a_h]\dots[a_g]$ already utilizes all the elements of $G_h$ it can, and adding a
    new element of the group cannot extend the LCS anymore. Hence, \AS{the LCS remains unchanged}.

    {\bf Case 1B.} Subproblems  due to  $\pi[1 \dots j-1]$: For some $i\in[g]$ let $\pi[j]\in G_i$.
    Then one candidate solution's size is DP$[j-1][a_1]\dots[a_i-1]\dots[a_g]+1$, where we take the
    string LCS$_{j-1}(a_1,\dots,a_i-1, \dots, a_g)$ and append $\pi[j]$ at the end of it. This
    increases the LCS value by 1 \AS{(on comparison of $\pi[1 \dots j]$ with, $\sigma_{j-1}
    (a_1,\dots,a_i-1, \dots, a_g)$ appended with $\pi[j]$)} to give a solution of size DP$[j-1]
    [a_1]\dots[a_i-1]\dots[a_g]+1$.

    Otherwise, if for some $h\in[g],\pi[j]\not\in G_h$, then by adding any element from $G_h$, we
    can never have it being the same as $\pi[j]$, and hence cannot extend the LCS any further.

    {\bf Case 2: From subproblems corresponding to smaller values of $j$.} Here we try to build a
    solution to DP$[j][a_1]\dots[a_i]\dots[a_g]$ by using the solution of DP$[j-1][a_1]\dots
    [a_i]\dots[a_g]$. In this case if the length of the LCS were to increase, then we note that the
    last element of the new LCS has to be $\pi[j]$ (otherwise, the longer LCS we find would also
    have been valid for DP$[j-1][a_1]\dots[a_i]\dots[a_g]$). So the solution to this problem can be
    obtained by appending $\sigma_{j-1}(a_1,\dots,a_i-1\dots a_g)$ (for the appropriate group $i\in
    [g]$) with $\pi[j]$ to get an LCS of length $1 + \text{DP}[j-1][a_1]\dots[a_i-1]\dots
    [a_g]$. Hence we note that this case essentially reduces to case 1B(note that this case hence
    does not feature in the recurrence, and is just mentioned for completeness).

    Therefore, by iterating over all groups to consider all possible candidates, we get the
    recurrence,
    \[ \text{DP}[j][a_1]\dots[a_g]=\max_{i\in[g]}
        \begin{cases}
            \text{DP}[j][a_1]\dots[a_i-1]\dots[a_g]
            + 1, \text{ if $\pi[1\dots j]$ has  $\ge a_i$ elements of
            $G_i$; } \\
            \text{DP}[j][a_1]\dots[a_i-1]\dots[a_g], \text{ if $\pi[1\dots j]$ has $\le a_i-1$
            elements of $G_i$.} \\
            \text{DP}[j-1][a_1]\dots[a_i-1]\dots[a_g] + 1,  \text{ if $\pi[j]$ is in $G_i$;} \\
            \text{DP}[j-1][a_1]\dots[a_i-1]\dots[a_g], \text{ if $\pi[j]$ is not in $G_i$};

        \end{cases}
    \]

    Also note that here we use the fact that the length of the LCS can only increase by one in any
    of the above cases. If the LCS increased by more than one, we can ignore one of the characters
    (being an arbitrary character, or the newly considered character from the prefix, corresponding
    to the appropriate case) of the newly obtained LCS, from its corresponding solution string
    $\sigma$, thereby obtaining a solution with a larger LCS for a previously solved subproblem.
    This would hence give us a contradiction.

    To ensure that only
    subproblems where the top-$\ell$ ranks satisfy the \AS{strict-}fairness constraints ($\forall i\in[g],
    \lfloor\alpha_i \ell\rfloor\le a_i\le\lceil\beta_i \ell\rceil$ and $\forall i\in[g],
    a_i\le|G_i|$) are used for construction, we set all `invalid' subproblems, to have a value of
    $-\infty$ (we do this only for prefixes of size larger than $k$).

    For the {\em base case}, we have that for any $j\in[d]$ and all groups $i\in[g],
    \text{DP}[j][0]\dots[a_i=1]\dots[0]=$1, if there is a element of $G_i$ in $\pi[1\dots j]$; and 0
    otherwise.
    Note that the above recurrence can be solved in a top-down approach, to fill up the DP table.
    For a formal description of the pseudocode of the algorithm, see Algorithm \ref{alg:UlamDP}.

    \AS{ Once we have filled up the DP table, we can consider all cells DP$[a_1]\dots[a_g]$ which
     satisfy fairness criteria(i.e. for the corresponding prefix of size $\gamma:=\sum_{i\in
     [g]} a_i$, satisfy the fairness constraints) and select a cell with the maximum LCS length. We
     construct the required closest \AS{strict}-fair ranking from the LCS corresponding to this cell as
     follows: first we backtrack to find out the subproblems from which this LCS string was
     obtained. We now know, for each intermediary step, the required number of elements from each
     group and also the elements in the LCS at that stage. So in increasing order of $\gamma$, we
     go over the intermediary subproblems, and if the required number of elements from any group
     is not already satisfied, use arbitrary elements of the group (not in LCS) and append them to
     the string being constructed. In this manner, we ensure that the LCS is still preserved, and
     the fairness criteria are satisfied, and the output is a valid ranking. }

    \AS{Let us now consider the runtime of this algorithm:} There are $\bigO(d^{g+1})$ DP
    subproblems, evaluating each of which takes $\bigO(d)$ time. Then, we iterate over all the
    possible valid fair strings $\sigma(a_1,\dots, a_g)$ (there are at most  $\bigO(d^{g})$ such
    strings) to find the one with the longest possible common subsequence. Constructing the
    optimal ranking from the LCS takes at most $O(d)$ time. This gives us an
    overall running time of $\bigO(d^{g+2})$.
\end{proof}

\section{Fair Rank Aggregation}
\label{sec:FRA}
We start this section by formally defining the \emph{fair rank aggregation} problem. Then we will
provide two meta-algorithms that approximate the fair aggregated ranking.
\begin{definition}[$q$-mean Rank Aggregation] \label{def:rank-aggregation}
    Consider a metric space $(\sym_d,\rho)$ for a $d \in \mathbb{N}$. Given an exponent parameter $q \in \mathbb{R}$, and a set $S\subseteq
    \sym_d$ of $n$ input rankings, the \emph{$q$-mean rank aggregation} problem asks to
    find a ranking $\sigma \in \sym_d$ (not necessarily from $S$) that minimizes the objective
    function
    $\obj_q(S,\sigma):= \left( \sum_{\pi\in S} \rho(\pi,\sigma)^q \right)^{1/q}$.

\end{definition}

Generalized mean or $q$-mean objective functions are well-studied in the context of clustering
\cite{ChlamtacMV22}, and division of goods \cite{barman21}.  We study it for the first time in the
context of rank aggregation.  For $q=1$, the above problem is also referred to as the \emph{median
ranking} problem or simply \emph{rank aggregation} problem~\cite{kemeny1959mathematics,
young1988condorcet, young1978consistent, dwork2001rank}. On the other hand, for $q =\infty$, the
problem is also referred to as the \emph{center ranking} problem or \emph{maximum rank aggregation}
problem~\cite{BACHMAIER20152, biedl2009complexity, popov2007multiple}. Both these special cases are
studied extensively in the literature with different distance measures, e.g., Kendall tau
distance~\cite{dwork2001rank, ACN08, kenyon2007rank, Schudy2012thesis, biedl2009complexity}, Ulam
distance~\cite{CDK21, BACHMAIER20152, chakraborty2021approximating}, Spearman footrule
distance~\cite{dwork2001rank, BACHMAIER20152}.

In the fair rank aggregation problem, we want the output aggregated rank to satisfy certain fairness
constraints.

\begin{definition}[$q$-mean Fair Rank Aggregation] \label{def:fair-rank-aggregation}
    Consider a metric space $(\sym_d,\rho)$ for a $d \in \mathbb{N}$.  Given an exponent parameter $q \in \mathbb{R}$, and a set $S\subseteq
    \sym_d$ of $n$ input rankings/permutations, the \emph{$q$-mean (\AS{block, strict})-fair rank
    aggregation} problem asks to find a (\AS{block, strict})-fair ranking $\sigma \in \sym_d$ (not necessarily
    from $S$) that minimizes the objective function
    $ \obj_q(S,\sigma):= \left( \sum_{\pi \in S} \rho(\pi ,\sigma)^q \right)^{1/q}.$
\end{definition}
It is worth noting that in the above definition, the minimization is over the set of all the (\AS{block, strict})-fair rankings in $\sym_d$. When clear from the context, we drop block/strict and refer to it as the $q$-mean
fair rank aggregation problem. Let $\sigma^*$ be a (\AS{block, strict})-fair ranking that minimizes
$\obj_q(S,\sigma)$,
then we call $\sigma^*$ a \emph{$q$-mean fair aggregated rank} of $S$. We refer to
$\obj_q(S,\sigma^*)$ as $\opt_q(S)$.

When $q=1$, we refer the problem as the \emph{fair median ranking} problem or simply \emph{fair rank
aggregation} problem. When $q=\infty$, the objective function becomes $\obj_\infty(S,\sigma) =
\max_{\pi\in S} \rho(\pi,\sigma)$, and we refer the problem as the \emph{fair center ranking}
problem.

Next, we present two meta algorithms that work for any values of $q$ and irrespective of \AS{strict, or block} fairness constraints.
We will also assume $q$ to be a constant.

\subsection{First Meta Algorithm}

\begin{theorem} \label{thm:fairmeta}
    Consider any $q \ge 1$. Suppose there is a $t(d)$-time $c$-approximation algorithm
    $\mathcal{A}$, for some $c \ge 1$, for the closest fair ranking problem over the metric space
    $(\sym_d,\rho)$. Then there exists a $(c+2)$-approximation algorithm for the $q$-mean fair rank
    aggregation problem, that runs in $\bigO(n \cdot t(d) + n^2\cdot f(d))$ time where $f(d)$ is the
    time to compute $\rho(\pi_1,\pi_2)$ for any $\pi_1,\pi_2 \in \sym_d$.
\end{theorem}
We devote this subsection to proving the above theorem. Let us start by describing the algorithm.
It works as follows: Given a set $S \subseteq \sym_d$ of rankings, it first computes $c$-approximate
closest fair ranking $\sigma$ (for some $c \ge  1$) for each $\pi \in S$. Next, among these $|S|$ fair rankings output the ranking $\sigma$
that minimizes $\obj_q(S,\sigma)$. Let us denote the output ranking by $\bar{\sigma}$. \AS{A formal
discription of the algorithm follows.}
\begin{algorithm}
    \caption{Meta-algorithm 1 for the $q$-mean fair rank aggregation.}
    \label{alg:fairmeta}
    \KwIn{A set $S \subseteq \sym_d$ of $n$ rankings.}
    \KwOut{A $(c+2)$-approximate fair aggregate rank of $S$.}
    Initialize $S' \leftarrow \emptyset$\\
    \For{each point $\pi$ in $S$ } {
        Find a $c$-approximate closest fair ranking $\sigma$ to $\pi$ using the algorithm
        $\mathcal{A}$\\
        $S' \leftarrow S' \cup \{\sigma\}$\\
    }
    Initialize $\bar{\sigma} \leftarrow \emptyset$\\
    Initialize $\obj_q(S,\sigma) \leftarrow \infty$\\
    \For{each point $\sigma$ in $S'$ } {
        \If{$\obj_q(S,\sigma)<\obj_q(S,\bar{\sigma})$}{
            $\bar{\sigma} \leftarrow \sigma$\\
        }
    }
    \Return{$\bar{\sigma}$ }
\end{algorithm}

It is straightforward to verify that the running time of the above algorithm is $\bigO(n \cdot t(d) +
n^2\cdot f(d))$, where $f(d)$ is the time to compute $\rho(\pi_1,\pi_2)$ for any $\pi_1,\pi_2 \in
\sym_d$ and $t(d)$ denotes the running time of the algorithm $\mathcal{A}$. So it only remains to
argue about the approximation factor of Algorithm~\ref{alg:fairmeta}. The following simple
observation plays a pivotal role in establishing the approximation factor of
Algorithm~\ref{alg:fairmeta}.
\begin{restatable}{lemma}{pointwise} \label{lem:pointwise}
    Given a set $S \subseteq \sym_d$ of $n$ rankings, let $\sigma^*$ be an optimal $q$-mean fair aggregated
    rank of $S$ under a distance function $\rho$. Further, let $\bar{\pi}$ be a nearest neighbor
    (closest ranking) of $\sigma^*$ in $S$, and $\bar{\sigma}$ be a $c$-approximate closest fair
    ranking to $\bar{\pi}$, for some $c \ge 1$. Then
    $\forall \pi\in S, \rho(\pi, \bar{\sigma}) \le (c+2)\cdot\rho(\pi, \sigma^*).$
\end{restatable}
\begin{proof}
    Since $\sigma^*$ is a fair ranking and $\bar{\sigma}$ is a $c$-approximate closest fair ranking
    to $\bar{\pi}$
    \begin{equation}
        \label{eq:approx-close}
        \rho(\bar{\pi}, \bar{\sigma}) \le c\cdot\rho(\bar{\pi}, \sigma^*).
    \end{equation}

    Since $\bar{\pi} \in S$ is a closest ranking to $\sigma^*$ in the set $S$,
    \begin{equation}
        \label{eq:closest-input}
        \forall \pi\in S, \rho(\bar{\pi}, \sigma^*) \le \rho(\pi, \sigma^*).
    \end{equation}
    Then it follows from~\autoref{eq:approx-close},
    \begin{equation}
        \label{eq:close-fair}
        \forall \pi\in S, \rho(\bar{\pi}, \bar{\sigma} ) \le c\cdot\rho(\pi, \sigma^*)
    \end{equation}

    Now, for any $\pi \in S$, we get,
    \begin{align*}
        \rho(\pi, \bar{\sigma}) &\le \rho(\pi,\sigma^*)+\rho(\sigma^*,\bar{\sigma})&&\text{(By the triangle inequality)}\\
                                &\le \rho(\pi,\sigma^*)+\rho(\sigma^*,\bar{\pi})+\rho(\bar{\pi},\bar{\sigma})&&\text{(By the triangle inequality)}\\
                                &\le \rho(\pi,\sigma^*)+\rho(\pi,\sigma^*)+c\cdot\rho(\pi,\sigma^*) && \text{(By~\autoref{eq:closest-input} and
                                ~\autoref{eq:close-fair})}\\
                                &\le (c+2)\cdot \rho(\pi, \sigma^*).
                                \qedhere
    \end{align*}
\end{proof}
Now, we use the above lemma to complete the proof of of \autoref{alg:fairmeta}.
\begin{proof}[Proof of \autoref{thm:fairmeta}.]
Let $\sigma^*$ be an (arbitrary)
optimal fair aggregate rank of $S$ and $\bar{\sigma}$ be the output of Algorithm~\ref{alg:fairmeta}.
The optimal value of the objective function is
\[ \opt=\obj_q(S,\sigma^*)=\left(\sum_{\pi\in S}\rho(\pi,\sigma^*)^q\right)^{1/q}.\]
Next, we show that $\obj_q(S,\bar{\sigma}) \le (c+2) \cdot
\opt$.
\begin{align*}
    \obj_q(S,\bar{\sigma})
    &= \left(\sum_{\pi\in S}\rho(\pi,\bar{\sigma})^q\right)^{1/q}\\
    &\le \left(\sum_{\pi\in S}\big((c+2)\cdot\rho(\pi,\sigma^*)\big)^q\right)^{1/q}\\
    &= (c+2)\cdot\left(\sum_{\pi\in S}\rho(\pi,\sigma^*)^q\right)^{1/q}\\
    &= (c+2)\cdot\opt.
\end{align*}
where the first inequality follows from \autoref{lem:pointwise}. This concludes the proof of~\autoref{thm:fairmeta}.
\end{proof}

\paragraph*{Applications of~\autoref{thm:fairmeta}.}

We have shown in \autoref{thm:StrongGreedy} that the closest \AS{block}-fair ranking problem for Kendall tau
can be solved exactly in $\bigO(d^2)$ time, i.e., the approximation ratio is $c=1$. We also know from
\cite{KTeval}, that the Kendall tau distance between two permutations can be computed in
$\bigO(d\log d)$ time. This gives us that,
\begin{corollary} \label{cor:KTAggr}
    For any $q\ge1$, there exists an $\bigO(nd^2+n^2d\log d)$ time meta-algorithm that finds a
    3-approximate solution to the $q$-mean \AS{block}-fair, rank aggregation problem under the Kendall tau
    metric.
\end{corollary}
It is shown in~\cite{CelisSV18} that the closest \AS{strict}-fair ranking problem for
Spearman footrule can be solved exactly in $\bigO(d^3\log d)$ time, i.e., the approximation ratio is
$c=1$. Since distance under Spearman footrule can be trivially computed in $\bigO(d)$ we have that,
\begin{corollary}  \label{cor:SFAggr}
    For any $q\ge1$, there exists an $\bigO(nd^3\log d+n^2d)$ time meta-algorithm that finds a
    3-approximate solution to the $q$-mean \AS{strict}-fair rank aggregation problem under the Spearman footrule metric.
\end{corollary}
We have shown in~\autoref{thm:ulamdp} that for constant number of groups, the closest \AS{strict}-fair ranking
problem for Ulam metric can be solved exactly in $\bigO(d^{g+2})$ time, i.e., the approximation
ratio is $c=1$. From~\cite{AD99} we know that Ulam distance between two permutations can be computed
in $\bigO(d\log d)$ time. This gives us that,
\begin{corollary} \label{cor:UlamAggr}
    For any $q\ge1$, there exists an $\bigO(nd^{g+2}+n^2d\log d)$ time meta-algorithm, that finds a
    3-approximate solution to the $q$-mean \AS{strict}-fair rank aggregation problem , under the Ulam metric.
\end{corollary}

We would like to emphasize that all the above results hold for any values of $q \ge 1$. Hence, they
are also true for the special case of the fair median problem (i.e., for $q=1$) and the fair center
problem (i.e., for $q = \infty$).

\subsection{Second Meta Algorithm}

\begin{theorem} \label{thm:fairmeta2}
    Consider any $q \ge 1$. Suppose there is a $t_1(n)$ time $c_1$-approximation algorithm
    $\mathcal{A}_1$ for some $c_1\ge1$ for $q$-mean rank aggregation problem; and a $t_2(d)$-time
    $c_2$-approximation algorithm $\mathcal{A}_2$, for some $c_2 \ge 1$, for the closest fair
    ranking problem over the metric space $(\sym_d,\rho)$. Then there exists a
    $(c_1c_2+c_1+c_2)$-approximation algorithm for the $q$-mean fair rank aggregation problem, that
    runs in $\bigO( t_1(n) + t_2(d) + n^2\cdot f(d))$ time where $f(d)$ is the time to compute
    $\rho(\pi_1,\pi_2)$ for any $\pi_1,\pi_2 \in \sym_d$.
\end{theorem}
The algorithm works as follows: Given a set $S \subseteq \sym_d$ of rankings, it first computes
$c_1$-approximate aggregate rank $\pi^*$. Next, output a $c_2$-approximate closest fair ranking
$\bar{\sigma}$, to $\pi^*$. \AS{A formal description of the algorithm follows.}
\begin{algorithm}
    \caption{Alternate Meta-algorithm for finding fair-aggregate rank.}
    \label{alg:fairmeta2}
    \KwIn{A set $S \subseteq \sym_d$ of $n$ rankings.}
    \KwOut{A $(c_1c_2+c_1+c_2)$-approximate fair aggregate rank of $S$.}
    Call $\mathcal{A}_1(S)$ to find the $c_1$-approximate aggregate rank $\pi^*$ of the input set
    $S$.\\
    Call $\mathcal{A}_2(\pi^*)$ to find the  $c_2$-approximate closest fair ranking $\bar{\sigma}$, to
    $\pi^*$.\\
    \Return{$\bar{\sigma}$}.
\end{algorithm}

It is easy to see that the running time of the algorithm is $\bigO( t_1(n) + t_2(d) + n^2\cdot
f(d))$, where $f(d)$ is the time to compute $\rho(\pi,\sigma)$ for any $\pi,\sigma \in \sym_d$,
$t_1(n)$ denotes the running time of the algorithm $\mathcal{A}_1$, and  $t_2(d)$ denotes the
running time of the algorithm $\mathcal{A}_2$. It now remains to argue about the approximation ratio
of the above algorithm. We again make a simple but crucial observation towards establishing the
approximation ratio for Algorithm~\ref{alg:fairmeta2}.

\begin{restatable}{lemma}{pointwisee} \label{lem:pointwise2}
    Given a set $S \subseteq \sym_d$ of $n$ rankings, let $\sigma^*$ be an optimal $q$-mean fair aggregated
    rank of $S$ under a distance function $\rho$. Further, let $\pi^*$ be the $c_1$-approximate
    aggregate rank of $S$ and $\bar{\sigma}$ be a $c_2$-approximate closest fair ranking to $\pi^*$,
    for some $c_1,c_1 \ge 1$. Then
    \[ \forall \pi\in S, \rho(\pi, \bar{\sigma}) \le (c_1c_1+c_1+c_2)\cdot\rho(\pi, \sigma^*). \]
\end{restatable}
\begin{proof}
    Since $\pi^*$ is a $c_1$-approximate rank aggregation of the input, we have that,
    \begin{equation} \label{eq:aggr}
        \rho(\pi, \pi^*) \le c_1\cdot\rho(\pi, \sigma^*).
    \end{equation}
    Since $\bar{\sigma}$ is the $c_2$-approximate closest fair rank to $\pi^*$, we have,
    \begin{equation} \label{eq:approx-close2}
        \rho(\pi^*,\bar{\sigma})\le c_2\cdot\rho(\pi^*,\sigma^*)
    \end{equation}

    So, for any $\pi\in S$ we get,
    \begin{align*}
        \rho(\pi, \bar{\sigma})
        &\le  \rho(\pi, \pi^*)+\rho(\pi^*, \bar{\sigma})  &&\text{(By the triangle inequality)}\\
        &\le  \rho(\pi, \pi^*)+c_2\cdot\rho(\pi^*,\sigma^*)
        &&\text{(By~\autoref{eq:approx-close2})}\\
        &\le  \rho(\pi, \pi^*)+
        c_2\cdot(\rho(\pi^*, \pi)+\rho(\pi, \sigma^*)) &&\text{(By the triangle inequality)}\\
                                                      &\le (1+c_2)\cdot \rho(\pi, \pi^*)+c_2\cdot\rho(\pi, \sigma^*) \\
                                                      &\le (1+c_2)\cdot c_1\cdot\rho(\pi, \sigma^*)+c_2\cdot\rho(\pi, \sigma^*)
                                                      &&\text{(By~\autoref{eq:aggr})} \\
                                                      &\le (c_1c_2+c_1+c_2)\cdot\rho(\pi, \sigma^*). \qedhere
    \end{align*}
\end{proof}
Once we have this key lemma in place, the
remaining proof of \autoref{thm:fairmeta2}, follows exactly as the proof of \autoref{thm:fairmeta}.

The above algorithm can give similar approximation guarantees as Algorithm~\ref{alg:fairmeta}, but
with potentially better running times depending on whether the rank aggregation problem is solved in
a faster way for the particular problem in consideration. For instance consider the case for
Spearman footrule. It is known that the rank aggregation problem for Spearman footrule can be solved
in $\Tilde{\bigO}(d^{2})$ time \cite{BrandLNPSS0W20}. So, using this in conjunction with
Algorithm~\ref{alg:fairmeta2} we obtain the following result.
\begin{restatable}{corollary}{SFmed} \label{cor:SFmed}
    For $q=1$, there exists an $\bigO(d^3\log d+n^2d+nd^2)$ time meta-algorithm, that finds a
    3-approximate solution to the $q$-mean \AS{strict}-fair rank aggregation problem (i.e., the fair median
    problem) under Spearman footrule metric.
\end{restatable}
\begin{proof}
In~\cite{dwork2001rank} it is shown that optimal rank aggregation under Spearman footrule can be
reduced in polynomial time to the minimum cost perfect matching in a bipartite graph. The reduction
takes $\bigO(nd^2)$ time, which is needed to compute the edge weights for the constructed bipartite
graph. Further, \cite{BrandLNPSS0W20} gives a randomized $\tilde{\bigO}(m+n^{1.5})$ time algorithm
for the minimum cost perfect bipartite matching problem for bipartite graphs with $n$ nodes and $m$
edges. The reduction in~\cite{dwork2001rank} creates an instance of the minimum cost perfect
bipartite matching problem with $O(d)$ nodes and $O(d^2)$ edges. Hence, the result of~\cite
{BrandLNPSS0W20} gives us an exact $\tilde{\bigO}(nd^2)$ time rank aggregation algorithm for
Spearman footrule, i.e., $c_1=1$ and $t_1=\tilde{\bigO}(n^2+d)$.

Since we can compute an exact closest \AS{strict}-fair ranking under Spearman footrule in $\bigO(d^3\log d)$,
we have that $c_2=1$, and $t_2=\bigO(d^3\log d)$. Moreover, we can also compute the Spearman footrule
distance in $f(d)=\bigO(d)$ time. Now plugging in these values to \autoref{thm:fairmeta2}, we get
that we can find a $3$-approximate solution to the \AS{strict}-fair median problem under Spearman footrule in $\bigO(d^3\log d+n^2d+nd^2)$ time.
\end{proof}

\subsection{Breaking below 3-factor for Ulam \label{sec:Ulambetter}}
In this section, we provide a polynomial time $(3-\varepsilon)$-approximation algorithm for the fair (\AS{for brevity we refer to strict-fairness as fairness in this section}) median problem under the Ulam metric for constantly many groups (proving~\autoref{thm:Ulambetter}). More specifically, we show the following theorem.
\begin{theorem} \label{thm:Ulambetter}
    For $q=1$, there exists a constant $\varepsilon > 0$ and an algorithm that, given a set $S \subseteq \sym_d$ of $n$ rankings of $d$ candidates where these candidates are partitioned into $g \ge 1$ groups, finds a $(3-\varepsilon)$-approximate solution to the $q$-mean fair rank aggregation problem (i.e., the fair median problem), under the Ulam metric in time $\bigO(nd^{g+2} + n^2 d \log d)$.
\end{theorem}

In proving the above theorem, we first describe the algorithm, then provide the running time analysis, and finally analyze the approximation factor.

\paragraph*{Description of the algorithm.}We run two procedures, each of which outputs a fair ranking (candidate fair median), and then return the one (among these two candidates) with the smaller objective value. The first procedure is the one used in Corollary~\ref{cor:UlamAggr}. The second procedure is based on the relative ordering of the candidates and is very similar to that used in~\cite{CDK21} with the only difference being that now we want the output ranking to be fair (which was not required in~\cite{CDK21}). Our second procedure Algorithm~\ref{alg:RelativeOrder} is given a parameter $\alpha \in [0, 1/10]$. It first creates a directed graph $H$ with the vertex set $V(H) = [d]$ and the edge set $E(H) = \{ (a,b) \mid a <_{\pi} b\text{ for at least }(1-2\alpha) n \text{ rankings }\pi \in S \}$. If this graph $H$ is acyclic, then we are done. Otherwise, we make $H$ acyclic as follows: Iterate over all the vertices $v \in V(H)$, and in each iteration, find a shortest cycle containing $v$, and then delete all the vertices (along with all the incident edges) in that cycle. In the end, we will be left with an acyclic subgraph $\bar{H}$ (of the initial graph $H$). Then, we perform a topological sorting to find an ordering $O$ of the vertex set $V(\bar{H})$. Note, $O$ need not be a ranking of $d$ candidates (since $V(\bar{H})$ could be a strict subset of $[d]$). Let $\bar{\sigma}_{par}$ be the sequence denoting the ordering $O$. Then we find a fair ranking $\bar{\sigma}$ that maximizes the length of an LCS with $\bar{\sigma}_{par}$ and output it.


\begin{algorithm}
    \caption{Relative Order Algorithm}
    \label{alg:RelativeOrder}
    \KwIn{A set $S \subseteq \sym_d$ of $n$ rankings, $\alpha \in [0,1/10]$.}
    \KwOut{A fair ranking from $\sym_d$.}
    $H \leftarrow ([d],E)$ where $E=\{ (a,b) \mid a <_{\pi} b\text{ for at least }(1-2\alpha) n \text{ rankings }\pi \in S \}$\\
    \For{each point $v \in V(H)=[d]$ } {
        $C \leftarrow $ A shortest cycle containing $v$\\
        $H = H - V(C)$\\
    }
    $\bar{H} \leftarrow H$\\
    $\bar{\sigma}_{par} \leftarrow $ Sequence representing a topological ordering of $V(\bar{H})$\\

    Find a fair ranking $\bar{\sigma} \in \sym_d$ that maximizes the length of an LCS with $\bar{\sigma}_{par}$\\

    \Return{$\bar{\sigma}$ }
\end{algorithm}

\paragraph*{Running time analysis.}The first procedure that is used in Corollary~\ref{cor:UlamAggr} takes $\bigO(nd^{g+2} + n^2 d \log d)$ time. It follows from~\cite{CDK21}, Step 1-6 of the second procedure (Algorithm~\ref{alg:RelativeOrder}) takes time $\bigO(nd^2 + d^3) = \bigO(nd^{g+2})$ (since $g \ge 1$). To perform Step 7 of Algorithm~\ref{alg:RelativeOrder}, we use the dynamic programming described in Algorithm~\ref{alg:UlamDP}. So, this step takes $\bigO(d^{g+2})$ time. As finally, we output the fair ranking produced by the two procedures that have a smaller objective value; the total running time is $\bigO(nd^{g+2} + n^2 d \log d)$.


\paragraph*{Showing the approximation guarantee.}Although the analysis proceeds in a way similar to that in~\cite{CDK21}, it differs significantly in many places since now the output of our algorithm must be a fair ranking (not an arbitrary one). Let $\sigma^*$ be an (arbitrary) fair median (1-mean fair aggregated rank) of $S$ under the Ulam metric. Then $\OPT(S) = \sum_{\pi \in S} \U(\pi,\sigma^*)$, which for brevity we denote by $\OPT$. For any $S' \subseteq S$, let $\OPT_{S'} := \sum_{\pi \in S'}\U(\pi,\sigma^*)$. Let us take parameters $\alpha,\beta,\gamma,\varepsilon,\eta,\nu$, the value of which will be set later. Let us consider the following set of fair rankings
\[
S_f := \{\sigma \in \sym_d \mid \sigma \text{ is a fair ranking closest to }\pi \in S\}.
\]

It is worth noting that our first procedure (the algorithm used in Corollary~\ref{cor:UlamAggr}) essentially outputs a fair ranking from the set $S_f$ with the smallest objective value.

Assume that
\begin{equation}
    \label{eq:ass-far}
    \forall_{\sigma \in S_f},\; \U(\sigma,\sigma^*) > (2-\varepsilon)\OPT/n.
\end{equation}
Otherwise, let $\sigma' \in S_f$ does not satisfy the above assumption, i.e., $\U(\sigma',\sigma^*) \le (2-\varepsilon)\OPT/n$. Then
\begin{align}
    \label{eq:ass-close}
    \sum_{\pi \in S} \U(\pi,\sigma') &\le \sum_{\pi \in S} (\U(\pi,\sigma^*) + \U(\sigma^* ,\sigma')) &&\text{(By the triangle inequality)} \nonumber\\
    &= \OPT + n \cdot \U(\sigma^* ,\sigma') \nonumber\\
    &\le (3-\varepsilon)\OPT.
\end{align}
So from now on, we assume~\ref{eq:ass-far}. It then follows from the triangle inequality that
\begin{equation}
    \label{eq:ass-far-input}
    \forall_{\pi \in S},\; \U(\pi,\sigma^*) > (1-\varepsilon/2)\OPT/n.
\end{equation}
To see the above, for the contradiction's sake, let us assume it is not true. Then for some $\pi \in S$, $\U(\pi,\sigma^*) \le (1-\varepsilon/2)\OPT/n$. Let $\sigma\in S_f$ be a closest fair ranking to $\pi$ in $S$. Since $\sigma^*$ is also a fair ranking, by the definition of $S_f$, $\U(\sigma,\pi) \le \U(\pi,\sigma^*)$. Hence $\U(\sigma,\sigma^*) \le \U(\sigma,\pi) + \U(\pi,\sigma^*) \le 2 \U(\pi,\sigma^*) \le (2-\varepsilon)\OPT/n$, contradicting Assumption~\ref{eq:ass-far}.

For each $\pi \in S \cup S_f$, consider an (arbitrary) LCS $\ell_\pi$ between $\pi$ and $\sigma^*$. Let $I_\pi$ denote the set of symbols in $[d]$ that are not included in $\ell_\pi$. Note, $|I_\pi| = d - |LCS(\pi , \sigma^*)| = \U(\pi,\sigma^*)$.

For any $\pi \in S \cup S_f$ and $D \subseteq [d]$, let $I_\pi(D) := I_\pi \cap D$. For any $a \in [d]$ and $S' \subseteq S$,
\[
\co_{S'}(a) := |\{\pi \in S' \mid a \text{ is not included in }\ell_\pi\}|.
\]
When $S'=S$, for brevity, we drop the subscript $S'$. For any $D \subseteq [d]$ and $S' \subseteq S$, $\OPT_{S'}(D):=\sum_{a \in D} \co_{S'}(a)$. When $D=[d]$, for brevity, we only use $\OPT_{S'}$. Note, $\OPT = \OPT_S$.

We call a symbol $a \in [d]$ \emph{lazy} if $\co(a) \le \alpha n$; otherwise \emph{active}. Let $L$ denote the set of all lazy symbols, and $A = [d] \setminus L$ (i.e., the set of all active symbols).

Now we divide our analysis into two cases depending on the size of $A$. Informally, in Case 1, the number of active symbols is small, and thus the objective value (or $\co$) is, in some sense, distributed over many symbols. As a consequence, for most pairs of symbols, we can retrieve the relative ordering between them (as in $\sigma^*$) just by looking into their relative ordering in "the majority" of the input rankings. However, there will be "a few" pairs of symbols for which we get the wrong relative ordering (because one of them might be an active symbol and thus moved for most input rankings with respect to $\sigma^*$). This may create cycles in the graph $H$. However, since there are only a small number of active symbols, we can remove cycles iteratively from the graph $H$ and reconstruct a large subsequence of $\sigma^*$, which is $\bar{\sigma}_{par}$. Hence any completion of that subsequence is close to $\sigma^*$ in the Ulam distance. Thus, the final output (fair) ranking $\bar{\sigma}$ acts as an approximate (fair) median. On the other hand, in Case 2, the number of active symbols is large, and thus the total $\co$ is large. Then we argue that one of the inputs' closest fair ranking is better than 3-approximate median.

\textbf{Case 1: $|A| \le \beta \cdot \OPT/n$}

We use the following result from~\cite{CDK21}.

\begin{claim}[\cite{CDK21}]
\label{clm:CDK}
The set of symbols in $L\cap V(\bar{H})$ forms a common subsequence between $\bar{\sigma}_{par}$ and $\sigma^*$. Furthermore, $|L \setminus V(\bar{H})| \le \frac{4 \alpha}{1 - 4 \alpha}|A|$.
\end{claim}
 Using the above claim, we show the following lemma.

 \begin{lemma}
 \label{lem:lazy-approx}
 Consider an $\alpha \in [0,1/10]$ and $\beta \in (0,1)$. Given an input set $S \subseteq \sym_d$ of size $n$, let the set of active symbols $A$ be of size at most $\beta \cdot \OPT/n$. Then on input $S, \alpha$, Algorithm~\ref{alg:RelativeOrder} outputs a fair ranking $\bar{\sigma}$ that is an $(1+3 \beta(1+8\alpha))$-approximate fair median (1-mean fair aggregated rank) of $S$.
 \end{lemma}
 \begin{proof}
First note, by~\autoref{clm:CDK}, the set of symbols in $L\cap V(\bar{H})$ forms a common subsequence between $\bar{\sigma}_{par}$ and $\sigma^*$. Now since $\sigma^*$ is a fair ranking, the length of an LCS between $\bar{\sigma}_{par}$ and $\bar{\sigma}$ must be at least $|L \cap V(\bar{H})| - |V(\bar{H}) \setminus L|$. Moreover,
\begin{equation}
    \label{eq:LCSfair}
    |LCS(\bar{\sigma},\sigma^*)| \ge |L \cap V(\bar{H})| - 2 |V(\bar{H}) \setminus L| \ge |L \cap V(\bar{H})| - 2 |A|.
\end{equation}
(Note, since a fair ranking $\sigma^*$ exists, it is always possible to perform Step 7 of Algorithm~\ref{alg:RelativeOrder} to output a fair ranking $\bar{\sigma}$.)


 \begin{align*}
     \obj(S,\bar{\sigma})&=\sum_{\pi \in S} \U(\pi , \bar{\sigma})\\
     &\le \sum_{\pi \in S} (\U(\pi, \sigma^*) + \U(\sigma^* , \bar{\sigma})) &&\text{(By the triangle inequality)}\\
     &= \OPT + n \cdot (d - |LCS(\sigma^* , \bar{\sigma})|)\\
     & \le \OPT + n \cdot (d - |L \cap V(\bar{H})| +  2|A|)&&\text{(By~\autoref{eq:LCSfair})}\\
     &= \OPT + n \cdot (|A| + |L| - |L \cap V(\bar{H})| + 2|A|)\\
     &= \OPT + n \cdot (3|A| + |L \setminus V(\bar{H})|)\\
     & \le \OPT + n \cdot \frac{3}{1-4 \alpha}|A|&&\text{(By~\autoref{clm:CDK})}\\
     & \le \OPT + \frac{3\beta}{1-4\alpha} \OPT &&\text{(Since } |A| \le \beta \cdot \OPT/n \text{)}\\
     &\le (1+3\beta(1+8 \alpha)) \OPT &&\text{(Since }\alpha\in [0,1/10]\text{)}.
 \end{align*}
 \end{proof}

\textbf{Case 2: $|A| > \beta \cdot \OPT/n$}

Recall, we consider parameters $\alpha,\beta,\gamma,\varepsilon,\eta,\nu$, the value of which will be set later. Let us partition $S$ into the following sets of \emph{near} and \emph{far} rankings: Let $N:=\{\pi \in S \mid \U(\pi,\sigma^*) < (1 + \varepsilon/\alpha) \OPT/n\}$, and $F:= S \setminus N$.
\begin{claim}[\cite{CDK21}]
\label{clm:CDK-high}
\begin{align}
    &\OPT_N \ge (1-\alpha/2)(1-\varepsilon/2)\OPT.\label{eq:near-opt}\\
    &\OPT_N(A) \ge \frac{\alpha n}{2} |A| \ge \frac{\alpha \beta}{2} \OPT_N\label{eq:near-opt-active}
\end{align}
\end{claim}
Let $R:= \{\pi \in N \mid |I_\pi(A)| \ge (1-\gamma)\frac{\alpha}{2} |A|\}$. Then
\begin{align}
    \label{eq:opt-non-R}
    \OPT_{N \setminus R}(A) &\le (1-\gamma) \frac{\alpha}{2} |A|\cdot |N \setminus R|\nonumber\\
    &\le (1-\gamma)\OPT_N(A)&&\text{(By~\autoref{eq:near-opt-active})}.
\end{align}
As a consequence, we get that
\begin{equation}
    \label{eq:opt-R}
    \OPT_R(A)\ge \gamma \OPT_N(A).
\end{equation}

Further partition $R$ into $R_1,\dots,R_r$ for $r=\lceil \log_{1+\nu} (2/(\alpha - \alpha \gamma)) \rceil $ as follows
\begin{equation}
    \label{eq:Ri}
    R_i := \{\pi \in R \mid (1+\nu)^{i-1} (1-\gamma) \frac{\alpha}{2} |A| < |I_\pi (A)| \le (1+\nu)^i (1-\gamma) \frac{\alpha}{2} |A|\}.
\end{equation}

By an averaging argument, there exists $i^* \in [r]$ such that for $R^* = R_{i^*}$,
\begin{equation}
    \label{eq:R*}
    \OPT_{R^*}(A) \ge \OPT_R(A)/r \ge \frac{\gamma}{r} \OPT_N(A)
\end{equation}
where the last inequality follows from~\autoref{eq:opt-R}.

Let us consider $\eta = \frac{1}{2} \left( (1+\nu)^{i^* - 1} (1-\gamma) \alpha/2 \right)^2$. For each $\pi \in S$, let $\sigma(\pi)$ be the closest fair ranking that is in the set $S_f$. Next, consider the following procedure to segregate $R^*$ into a set of clusters with $\mathcal{C}$ denoting the set of cluster centers. (We emphasize that the following clustering step is only for the sake of analysis.)
\begin{enumerate}
    \item Initialize $\mathcal{C} \leftarrow \emptyset$.
    \item Iterate over all $\pi \in R^*$ (in the non-decreasing order of $|I_\pi(L)|$)
    \begin{enumerate}
        \item If for all $\pi' \in \mathcal{C}$, $|I_{\sigma(\pi')} \cap I_\pi (A)| < \eta |A|$, then add $\pi$ in $\mathcal{C}$. Also, create a cluster $C_\pi \leftarrow \{ \pi \}$.
        \item Else pick some $\pi' \in \mathcal{C}$ (arbitrarily) such that $|I_{\sigma(\pi')} \cap I_\pi (A)| \ge \eta |A|$ and add $\pi$ in $C_{\pi'}$.
    \end{enumerate}
\end{enumerate}

Since we process all $\pi \in R^*$ in the non-decreasing order of $|I_\pi(L)|$, it is straightforward to see
\begin{equation}
    \label{eq:cluster-non-decrease}
    \forall_{\pi' \in \mathcal{C}},\; \forall_{\pi \in C_{\pi'}},\; |I_{\pi'}(L)| \le |I_{\pi}(L)|.
\end{equation}

Next, we use the following simple combinatorial lemma from~\cite{CDK21}.
\begin{lemma}[\cite{CDK21}]
\label{lem:set-intersect}
For any $c,d \in \mathbb{N}$ and $\eta \in (0,\frac{1}{2c^2}]$, any family $\mathcal{F}$ of subsets of $[d]$ where
\begin{itemize}
    \item Every subset $I \in \mathcal{F}$ has size at least $n/c$, and
    \item For any two $I \ne J \in \mathcal{F}$, $|I \cap J| \le \eta d$,
\end{itemize}
has size at most $2c$.
\end{lemma}

We use the above lemma to prove the following.
\begin{claim}
\label{clm:number-cluster}
For any $\alpha \in (0,1)$ and $\gamma \in (0,0.5)$, $|\mathcal{C}| \le 8/\alpha$.
\end{claim}
\begin{proof}
Recall, by definition, for any $\pi' \in \mathcal{C}$, $\sigma(\pi') \in S_f$. By Assumption~\ref{eq:ass-far}, for any $\sigma \in S_f$,
\[
|I_{\sigma}| > (2-\varepsilon) \OPT/n \ge (2-\varepsilon)\alpha |A|
\]
where the last inequality follows since by the definition of active symbols, $\OPT \ge \alpha n |A|$. Also, by~\autoref{eq:Ri}, for each $\pi \in R^*$, $|I_\pi (A)| \ge (1+\nu)^{i^*-1} (1-\gamma) \frac{\alpha}{2} |A|$. Then it directly follows from~\autoref{lem:set-intersect}, $|\mathcal{C}|\le 2 \left \lceil \frac{1}{(1+\nu)^{i^*-1} (1-\gamma) \frac{\alpha}{2}}\right \rceil \le 8/\alpha$ (for $\gamma < 0.5$).
\end{proof}

Since $C_{\pi'}$'s create a partitioning of the set $R^*$, by averaging,
\begin{align}
    \label{eq:large-cluster-A}
    \exists_{\pi' \in \mathcal{C}},\; \OPT_{C_{\pi'}}(A) & \ge \frac{\OPT_{R^*}(A)}{ |\mathcal{C}| }\nonumber\\
    & \ge \frac{\alpha \gamma}{8r} \OPT_N(A)&&\text{(By~\autoref{eq:R*} and~\autoref{clm:number-cluster})}.
\end{align}

From now on, fix a $\tilde{\pi} \in \mathcal{C}$ that satisfies the above. Then,
\begin{align}
    \label{eq:large-cluster}
    \OPT_{C_{\tilde{\pi}}} \ge \OPT_{C_{\tilde{\pi}}} (A) &\ge \frac{\alpha \gamma}{8r} \OPT_N (A)&&\text{(By~\autoref{eq:large-cluster-A})}\nonumber\\
    &\ge \frac{\alpha^2 \beta \gamma}{16r} \OPT_N&&\text{(By~\autoref{eq:near-opt-active})}.
\end{align}

Note, $\OPT_{C_{\tilde{\pi}}}(L) + \OPT_{C_{\tilde{\pi}}}(A) = \OPT_{C_{\tilde{\pi}}} \le \OPT_N$, and thus,
\begin{equation}
    \label{eq:large-cluster-G}
    \OPT_{C_{\tilde{\pi}}} (L) \le \left( 1 - \frac{\alpha^2 \beta \gamma}{16r} \right) \OPT_{C_{\tilde{\pi}}}.
\end{equation}

Let $\tilde{\sigma} = \sigma(\tilde{\pi})$. Observe, since $\U(\tilde{\pi} , \tilde{\sigma}) \le \U(\tilde{\pi} , \sigma^*)$ (recall, $\sigma^*$ is a fair ranking and $\tilde{\sigma}$ is a closest fair ranking to $\tilde{\pi}$), by the triangle inequality,
\begin{equation}
    \label{eq:close-fair-tilde}
    \U(\sigma^* , \tilde{\sigma}) \le 2 \U(\sigma^* , \tilde{\pi}).
\end{equation}

Next, take any far ranking $\pi \in F$.
\begin{align*}
    \U(\pi, \tilde{\sigma}) &\le \U(\pi , \sigma^*) + \U(\sigma^*, \tilde{\sigma})&&\text{(By the triangle inequality)}\\
    &\le \U(\pi,\sigma^*) + 2  \U(\sigma^* , \tilde{\pi})&&\text{(By~\autoref{eq:close-fair-tilde})}\\
    &\le 3 \U(\pi ,\sigma^*)&&\text{(Since } \pi \in F \text{ and }\tilde{\pi} \in N \text{)}.
\end{align*}
Hence, we get that
\begin{equation}
    \label{eq:approx-far}
    \forall_{\pi \in F},\; \U(\pi, \tilde{\sigma}) \le 3 \U(\pi,\sigma^*).
\end{equation}

Then consider any near ranking $\pi \in N$.
\begin{align*}
    \U(\pi,\tilde{\sigma}) & \le \U(\pi , \sigma^*) + \U(\sigma^*, \tilde{\sigma})&&\text{(By the triangle inequality)}\\
    &\le \U(\pi,\sigma^*) + 2  \U(\sigma^* , \tilde{\pi})&&\text{(By~\autoref{eq:close-fair-tilde})}\\
    &\le \U(\pi,\sigma^*) + 2 (1+\varepsilon/\alpha)\OPT/n &&\text{(Since }\tilde{\pi} \in N\text{)}\\
    &\le \U(\pi,\sigma^*) + 2 \frac{1+\varepsilon/\alpha}{1-\varepsilon/2} \U(\pi,\sigma^*)&&\text{(By Assumption~\ref{eq:ass-far-input})}\\
    &\le \frac{3 + (4/\alpha - 1) \varepsilon/2}{1-\varepsilon/2} \U(\pi,\sigma^*).
\end{align*}

Hence, we get that
\begin{equation}
    \label{eq:approx-near}
    \forall_{\pi \in N},\; \U(\pi, \tilde{\sigma}) \le \frac{3 + (4/\alpha - 1) \varepsilon/2}{1-\varepsilon/2} \U(\pi,\sigma^*).
\end{equation}

Lastly, consider any $\pi \in C_{\tilde{\pi}}$. Note,
\begin{align*}
    \U(\pi,\tilde{\sigma}) &\le  |I_\pi| + |I_{\tilde{\sigma}}| - |I_\pi \cap I_{\tilde{\sigma}}|\\
    &\le |I_\pi| + |I_{\tilde{\sigma}}| - |I_\pi(A) \cap I_{\tilde{\sigma}}|\\
    &\le |I_\pi| + |I_{\tilde{\sigma}}| - \eta |A|&&\text{(By the construction)}.
\end{align*}
Hence,
\begin{align}
    \label{eq:sum-cluster}
    \sum_{\pi \in C_{\tilde{\pi}}} \U(\pi,\tilde{\sigma}) & \le  \sum_{\pi \in C_{\tilde{\pi}}} (|I_\pi| + |I_{\tilde{\sigma}}| - \eta |A|) \nonumber\\
    & \le \sum_{\pi \in C_{\tilde{\pi}}} (|I_\pi| + 2 |I_{\tilde{\pi}}| - \eta |A|)&&\text{(By~\autoref{eq:close-fair-tilde})} \nonumber\\
    & \le \sum_{\pi \in C_{\tilde{\pi}}} (|I_\pi(L)| + 2 |I_{\tilde{\pi}}(L)|) + \sum_{\pi \in C_{\tilde{\pi}}} (|I_\pi(A)| + 2 |I_{\tilde{\pi}}(A)| - \eta |A|)\nonumber\\
    & \le 3 \sum_{\pi \in C_{\tilde{\pi}}} |I_\pi(L)| + \sum_{\pi \in C_{\tilde{\pi}}} (|I_\pi(A)| + 2 |I_{\tilde{\pi}}(A)| - \eta |A|)&&\text{(By~\autoref{eq:cluster-non-decrease})}\nonumber\\
    & \le 3 \sum_{\pi \in C_{\tilde{\pi}}} |I_\pi(L)| + \sum_{\pi \in C_{\tilde{\pi}}} (|I_\pi(A)| + 2 (1+\nu) |I_\pi(A)| - \eta |A|)&&\text{(Since }\pi,\tilde{\pi} \in R^*\text{)}\nonumber\\
    & \le 3 \sum_{\pi \in C_{\tilde{\pi}}} |I_\pi| - \sum_{\pi \in C_{\tilde{\pi}}} (\eta |A| - 2 \nu |I_\pi(A)|) \nonumber\\
    & \le 3 \sum_{\pi \in C_{\tilde{\pi}}} |I_\pi| - \left(\frac{2 \eta}{\alpha(1+\nu)^{i^*}(1-\gamma)} - 2 \nu \right) \sum_{\pi \in C_{\tilde{\pi}}}|I_\pi(A)| &&\text{(Since }\pi \in R^*\text{)}\nonumber\\
    & = 3 \sum_{\pi \in C_{\tilde{\pi}}} |I_\pi| - \left(\frac{(1+\nu)^{i^*-2} (1-\gamma) \alpha}{4} - 2 \nu \right) \sum_{\pi \in C_{\tilde{\pi}}}|I_\pi(A)|&&\text{(Recall, }\eta = \frac{( (1+\nu)^{i^* - 1} (1-\gamma) \alpha/2 )^2}{2} \text{)}\nonumber\\
    &= 3 \sum_{\pi \in C_{\tilde{\pi}}} |I_\pi| - \rho \sum_{\pi \in C_{\tilde{\pi}}}|I_\pi(A)|&&\text{(Let }\rho = \frac{(1+\nu)^{i^*-2} (1-\gamma) \alpha}{4} - 2 \nu \text{)}\nonumber\\
    & \le (3-\rho) \sum_{\pi \in C_{\tilde{\pi}}} |I_\pi| + \rho \sum_{\pi \in C_{\tilde{\pi}}} |I_\pi(G)|\nonumber\\
    & = (3-\rho) \OPT_{C_{\tilde{\pi}}} + \rho \OPT_{C_{\tilde{\pi}}} (G)&&\text{(By the definition)}\nonumber\\
    & \le \left( 3 - \frac{\rho \alpha^2 \beta \gamma}{16r} \right) \OPT_{C_{\tilde{\pi}}} &&\text{(By~\autoref{eq:large-cluster-G})}.
\end{align}

Finally, we deduce that
\begin{align*}
    \sum_{\pi \in S} \U(\pi , \tilde{\sigma}) &= \sum_{\pi \in F} \U(\pi , \tilde{\sigma}) + \sum_{\pi \in N \setminus C_{\tilde{\pi}}} \U(\pi , \tilde{\sigma}) + \sum_{\pi \in C_{\tilde{\pi}} }\U(\pi , \tilde{\sigma})\\
    & \le 3 \OPT_F + \frac{3 + (4/\alpha - 1) \varepsilon/2}{1-\varepsilon/2} \OPT_{N\setminus C_{\tilde{\pi}} } + \left( 3 - \frac{\rho \alpha^2 \beta \gamma}{16r} \right) \OPT_{C_{\tilde{\pi}}}&&\text{(By Equations~\ref{eq:approx-far},~\ref{eq:approx-near},~\ref{eq:sum-cluster})}\\
    &\le 3 \OPT + (3 + 2/\alpha + (4/\alpha - 1) \varepsilon/2) \varepsilon \OPT_{N\setminus C_{\tilde{\pi}} } - \frac{\rho \alpha^2 \beta \gamma}{16r} \OPT_{C_{\tilde{\pi}}} &&\text{(Assuming $\varepsilon < 0.5$)}\\
    &\le 3 \OPT + (3 + 2/\alpha + (4/\alpha - 1) \varepsilon/2) \varepsilon \OPT_{N} - \frac{\rho \alpha^2 \beta \gamma}{16r} \OPT_{C_{\tilde{\pi}}}\\
    & \le 3 \OPT - \left(\frac{\rho \alpha^4 \beta^2 \gamma^2}{256 r^2} - (3 + 2/\alpha + (4/\alpha - 1) \varepsilon/2) \varepsilon \right) \OPT_N &&\text{(By~\autoref{eq:large-cluster})}\\
    & \le \left( 3 - (1-\alpha/2)(1-\varepsilon/2) (\frac{\rho \alpha^4 \beta^2 \gamma^2}{256 r^2} - (3 + 2/\alpha + (4/\alpha - 1) \varepsilon/2) \varepsilon ) \right) \OPT&&\text{(By~\autoref{eq:near-opt})}.
\end{align*}

Consider an $\alpha \in (0,1/10]$ and $\beta \in (0,1)$. Then set $\gamma = 1/4$, $\nu = (1-\gamma)\alpha/32$. Now, it is not hard to verify that the constant $(1-\alpha/2)(1-\varepsilon/2) (\frac{\rho \alpha^4 \beta^2 \gamma^2}{256 r^2} - (3 + 2/\alpha + (4/\alpha - 1) \varepsilon/2) \varepsilon ) > 0$ for a proper choice of $\varepsilon < 0.5$ (which is a function of $\alpha,\beta$). Hence, we conclude the following.

 \begin{lemma}
 \label{lem:active-approx}
 Consider an $\alpha \in (0,1/10]$ and $\beta \in (0,1)$. There exists a constant $\delta > 0$ such that given an input set $S \subseteq \sym_d$ of size $n$, for which the set of active symbols $A$ is of size at least $\beta \cdot \OPT/n$, there exists a fair ranking $\tilde{\sigma} \in S_f$ such that $\obj_1(S,\tilde{\sigma}) \le (3-\delta) \OPT$.
 \end{lemma}

Recall our first procedure (the algorithm used in Corollary~\ref{cor:UlamAggr}) essentially outputs a fair ranking from the set $S_f$ with the smallest objective value. Let us set $\alpha = 1/10$ and $\beta = 1/6$. Now, the approximation guarantee of~\autoref{thm:Ulambetter} follows from~\autoref{lem:lazy-approx} and~\autoref{lem:active-approx}.

\section{Conclusion}
In this paper, we lay the theoretical framework for studying the closest fair ranking and fair rank
aggregation problems while ensuring minority protection and restricted dominance. We first give a
simple, practical, and exact algorithm for CFR under the Kendall tau metric; and a polynomial time
exact algorithm for CFR under the Ulam metric when there are constantly many groups. We then use
such a black box solution to CFR to design two novel meta-algorithms for FRA for a general $q$-mean
objective, which are valid under any metric. The approximation ratios of these meta algorithms
depend on the approximation ratio of the CFR subroutine used by them. Lastly, we give a
$(3-\eps)$-approximate algorithm for FRA under the Ulam metric, which improves over our meta
algorithm's approximation ratio for the same case. Achieving a similar (better than 3-factor)
approximation bound for the Kendall tau or other metrics is an interesting open problem. Another set of intriguing
open problems arises when there is overlap between the groups, i.e., when an element  can belong to multiple groups.
It is still open  whether exact algorithms exist for the
the case when an element can belong to (at most) two groups.

\paragraph{Acknowledgments. }
Diptarka Chakraborty was supported in part by an MoE AcRF Tier 2 grant (MOE-T2EP20221-0009) and an NUS ODPRT grant (WBS No. A-0008078-00-00).
Arindam Khan gratefully acknowledges the generous support due to Pratiksha Trust Young
Investigator Award,  Google India Research Award, and Google ExploreCS Award.

\bibliographystyle{alpha}
\bibliography{ref_fairrank}

\end{document}